\documentclass[11pt]{book}
\usepackage{iiit_thesis}
\usepackage[section]{placeins}
\pdfoutput=1
\usepackage{times}
\usepackage{epsfig}
\usepackage{amsmath}
\usepackage{amsthm}
\usepackage{amssymb}
\usepackage{latexsym}
\usepackage{graphicx}
\usepackage{multirow}
\usepackage{algorithm2e}
\usepackage{caption}
\usepackage[noadjust]{cite}
\graphicspath{{images/}{./}} 

\long\def\symbolfootnote[#1]#2{\begingroup%
\def\thefootnote{\fnsymbol{footnote}}\footnote[#1]{#2}\endgroup}
\renewcommand{\baselinestretch}{1.2}
\onecolumn
\renewcommand{\vec}[1]{\mathbf{#1}}
\renewcommand{\hat}[1]{\oldhat{\mathbf{#1}}}
\newcommand\ddfrac[2]{\frac{\displaystyle #1}{\displaystyle #2}}
\theoremstyle{definition}
\newtheorem{theorem}{Theorem}
\newtheorem{remark}{Remark}
\newtheorem{assum}{Assumption}

\DeclareMathOperator{\sgn}{sgn}
\DeclareMathOperator{\sat}{sat}
\DeclareMathOperator{\diag}{diag}
\DeclareSymbolFont{newfont}{OML}{cmm}{m}{it}
\DeclareMathSymbol{\Varrho}{3}{newfont}{37}

\begin{document}
\pagenumbering{roman}

\thispagestyle{empty}
\begin{center}
\vspace*{1.5cm}
{\Large \bf Robust Trajectory Tracking and Payload Delivery of a Quadrotor Under Multiple State Constraints}

\vspace*{3.75cm}
{\large Thesis submitted in partial fulfillment\\}
{\large  of the requirements for the degree of \\}

\vspace*{1cm}
{\it {\large Master of Science} \\
{\large in\\}
{\large Electronics and Communication Engineering by Research\\}}

\vspace*{1cm}
{\large by}

\vspace*{5mm}
{\large SOURISH GANGULY\\}
{\large 2019702015\\
{\small \tt sourish.ganguly@research.iiit.ac.in}}

\vspace*{4.0cm}
\includegraphics[width=14mm]{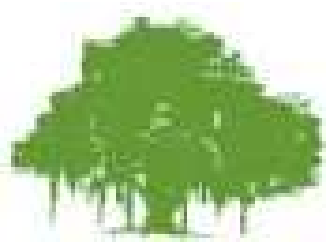}\\
{\large ROBOTICS RESEARCH CENTER \\}
{\large International Institute of Information Technology\\}
{\large Hyderabad - 500 032, INDIA\\}
{\large OCTOBER 2021\\}
\end{center}

\newpage
\thispagestyle{empty}
\renewcommand{\thesisdedication}{{\large Copyright \copyright~~Sourish Ganguly, 2021\\}{\large All Rights Reserved\\}}
\thesisdedicationpage

\newpage
\thispagestyle{empty}
\vspace*{1.5cm}
\begin{center}
{\Large International Institute of Information Technology\\}
{\Large Hyderabad, India\\}
\vspace*{3cm}
{\Large \bf CERTIFICATE\\}
\vspace*{1cm}
\noindent
\end{center}
It is certified that the work contained in this thesis, titled ``Robust Trajectory Tracking and Payload Delivery of a Quadrotor Under Multiple State Constraints'' by Sourish Ganguly, has been carried out under
my supervision and is not submitted elsewhere for a degree.

\vspace*{3cm}
\begin{tabular}{cc}
\underline{\makebox[1in]{}} & \hspace*{5cm} \underline{\makebox[2.5in]{}} \\
Date & \hspace*{5cm} Advisor: Dr. Spandan Roy
\end{tabular}
\oneandhalfspace

\newpage
\thispagestyle{empty}
\renewcommand{\thesisdedication}{\large To MY PARENTS}
\thesisdedicationpage

\mastersthesis
\renewcommand{\baselinestretch}{1.5}

\chapter*{Acknowledgments}
\label{ch:ack}
First and foremost, I wish to express my sincerest gratitude to my advisor, Dr. Spandan Roy, for his endless support throughout my years of research at IIIT, Hyderabad. Because of his guidance, I have been able to vastly improve my outlook towards tackling challenging problems. He has always been extremely kind and approachable whenever I had any doubts regarding the subject matter. In this two years, I have greatly enjoyed learning and re-learning intricacies of control theory, as he has always mentored me in such a way where I could have space to think independently and confidently. My progress and growth as a researcher is because of his constant encouragement.

Secondly, I want to thank the Robotics Research Center (RRC). I wish to thank every senior and every research scholar who never made me forget how exciting research life is. Their positive attitude towards tackling problem statements constantly gave me strength to progress. 

I feel immensely privileged to have been a part of IIIT Hyderabad. I have made great friends here who have always been a source of indispensable support. The seniors have always been extremely kind, and I have learnt a great deal from them. I have thoroughly enjoyed every second of my time here, and I will try my best to carry its core values and principles for the rest of my life.

I wish to thank all my professors throughout my journey as an engineering student who always been extremely supportive of my decision to pursue a career in research, and my school teachers who, even to this day, encourage me to keep pushing towards my goals.

I cannot thank my family enough for supporting my dreams. My journey to research would have never reached its destination had it not been for the unfathomable amount of faith my family has in me. Throughout my ups and downs, my mother has always been my strongest pillar, reminding me of my strength and my potential. I have always been inspired by my father's determination and focus, and he has greatly helped me in developing self-discipline and composure.

Lastly, I want to thank all my friends who have been there for me through countless years of us growing up together. I have always found their presence during times of difficulty, and, whenever I have felt unsure, they have cheered me on and helped me restore faith in myself. I am extremely fortunate to have them by my side. 

\chapter*{Abstract}
\label{ch:abstract}
With quadrotors becoming immensely popular in applications such as relief operations, infrastructure maintenance etc., a key control design challenge arises when the quadrotor has to manoeuvre through constrained spaces during various operational scenarios: for example, inspecting a pipeline within predefined velocity and space, dropping relief material at a precise location under tight spaces etc., under the face of parametric uncertainties and external disturbances. To tackle such scenarios, a controller needs to \textit{ensure a predefined tracking accuracy so as not to violate the constraints while simultaneously tackling uncertainties and disturbances}. However, state-of-the-art controllers dealing with constrained system motion are either not applicable for an underactuated system like quadrotor, or cannot tackle system uncertainties under full state constraints. This work attempts to fill such a gap in literature by designing Barrier Lyapunov Function (BLF) based robust controllers to satisfy multiple state-constraints while simultaneously negotiating parametric uncertainties and external disturbances. The superiority of the BLF control method over a typical unconstrained controller is demonstrated, followed by a robust control design to satisfy position and orientation constraints on quadrotor dynamics. Finally, full state-constraints on a quadrotor(i.e., constraints on the position, orientation, linear velocity and angular velocity) are satisfied with robust control. For each control design, the closed-loop system stability is studied analytically and the efficacy of the design is validated extensively either via realistic simulation scenarios or via experiments performed on a real quadrotor.

\tableofcontents
\listoffigures
\listoftables

\chapter{Introduction}
\label{ch:intro}

\renewcommand{\vec}[1]{\mathbf{#1}}
\renewcommand{\hat}[1]{\oldhat{\mathbf{#1}}}

Quadrotors have become increasingly popular in the research and the commercial domain, owing to their simple mechanical structure, vertical take-off and landing ability, and energy efficiency. They are relatively inexpensive and simple in implementation in the domain of Unmanned Aerial Vehicles (UAV), and so, the recent years saw a rapid growth and applicability in various domains. For example, quadrotors are extensively used nowadays in agricultural applications \cite{zhang2018agriculture, rao2019pesticide}, surveillance operations \cite{park2019surveillance, fari2019addressing}, wildlife and forest monitoring \cite{Torabi2018wildlife, wang2020problem}, search and rescue operations \cite{ghazali2016maritime, invernizzi2019dynamic, invernizzi2019integral}, etc. However, the inherent underactuated nature of a quadrotor poses considerable challenge to the Control researchers \cite{Ahmad2008underactuated, roy2021adaptive, baldi2020towards}. Additionally, quadrotors have fast-changing dynamics which imposes practical constraints on the control design in constrained spaces and under parametric uncertainties \cite{Tal2020Aggressive, roy2019role}. Therefore, a desirable objective is to develop control framework to tackle uncertainties while allowing a quadrotor to manoeuvre in various operational scenarios.

In the following sections of this introductory chapter, the motivation for this thesis is discussed, followed by a brief overview of the relevant fundamental principles and methods involved in constructing this work. Then, a brief section on the contribution of this thesis is provided, followed by the overall organization of the thesis chapters.

\section{Motivation} 
\label{sec12:Motivation}

\subsection{Quadrotors in rescue operations and payload transportation}

Out of the several applications of UAVs, the chief motivation behind this thesis is the increasing usage of quadrotors in critical missions such as rescue and relief operations and efficient package delivery during disaster management where it has to face several key challenges during autonomous manoeuvring. To illustrate, consider the following scenario:

\begin{figure*}[ht!]
    \includegraphics[width=1\textwidth, height=3.5in]{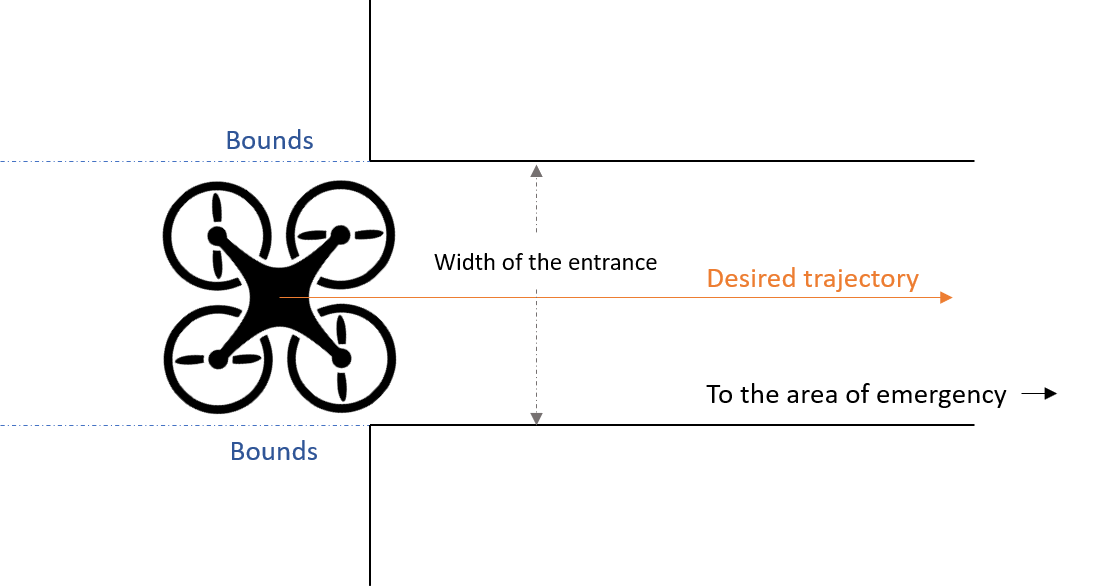}
    \centering
    \caption{A scenario where the quadrotor has to autonomously manoeuvre through a narrow opening in an emergency operation}
     \label{fig:quad_scenario}
\end{figure*}
Fig. \ref{fig:quad_scenario} describes a scenario where a quadrotor (top view) is commanded to follow a trajectory that proceeds through a narrow opening that leads to an emergency area. In this situation:

\begin{itemize}
\item The quadrotor has to successfully manoeuvre through the narrow space. It might also encounter other physically constrained areas such as open windows, pipelines, etc.
\item If the quadrotor is carrying a payload whose exact mass is not completely known (for example: ration supply), it has to deliver it to the designated spot with high precision.
\end{itemize}

This scenario demands the quadrotor to operate under strict physical constraints which limit exercising free movement at any instant. To illustrate, the quadrotor in Fig. \ref{fig:quad_scenario} passing through the narrow opening cannot deviate more than the distance between the boundary wall and the edge of a quadrotor frame, measured when the centre of the quadrotor frame lies on the desired trajectory. Similar constraints in the other directions also apply (which are not decipherable from the diagram) where the quadrotor needs to always maintain a tracking error less than what would result in collision with the boundary walls. Additionally, while dropping the payload at the target, the quadrotor cannot deviate from the allowable bounds of the target. To summarize, constraint handling is critical not only to the payload delivery operation, but to the safety of the quadrotor itself.

These critical operations have been explored in various works: vision based SLAM technology is employed in \cite{Bera2019Narrow} for autonomous navigation of a quadrotor; the work \cite{Su2019SwingConstraint} details optimal trajectory generation of a quadrotor where the payload has a constraint on its swing angle. 
Controlling the quadrotor under such critical scenarios having \textit{space constraints} requires precise tracking with \textit{predefined accuracy}, which becomes even more challenging if the knowledge on the system dynamic parameters are limited (such as the variability in the payload mass and center-of-mass (COM)) and there are significant external disturbances (such as wind gust, etc). In the following section, related works and contributions regarding the above issues are explored.


\subsection{Related works and challenges}
\label{subsection: Related works}

It is known that the states of a plant undergo transients. Some of the key parameters of these transients include the rise-time, settling time, overshoot, etc. These parameters are usually controlled via tuning of the user-defined control gains. In purely theoretical systems, where the system parameters are completely known, tuning is a definitive approach to mitigate undesirable transients. Research has progressed to robust and adaptive control solutions which rely on tuning to ensure user-specific desirable performance of a stable system. However, it is shown in \cite{KPTee2009BLF, dhar2021indirect} that, in absence of any dedicated design parameter that can explicitly impose a restriction on the maximum deviation of the state from its desired quantity, conventional robust or adaptive control solutions require extensive tuning to approach towards any user-specified performance limit. In other words, tuning requires extensive trial runs in scenarios where physical constraints become key challenges to the tracking problem. This defeats the purpose of the quadrotor having a relatively quicker response to an emergency situation. This even becomes more challenging in the face of various uncertainties, such as limited knowledge on dynamic parameters and payload. Therefore, to ensure manoeuvring of a quadrotor under limited space with a predefined accuracy, one has to look beyond the state-of-the-art robust \cite{xu2008sliding, sanchez2012continuous, derafa2012super, madani2007sliding, roy2019new} and adaptive \cite{nicol2011robust, bialy2013lyapunov, roy2020robust, dydek2012adaptive, tran2018adaptive, roy2019simultaneous,  tian2019adaptive} control solutions. 

Usually, there are two popular approaches towards tackling such constraint handling problems:

\begin{itemize}
    \item Barrier Lyapunov function (BLF) based solutions  \cite{KPTee2009BLF, liu2016barrier,dasgupta2019singularity,sachan2019output}, and
    \item Model Predictive Control (MPC) method \cite{dhar2021indirect,berberich2020data,allgower2012nonlinear}
\end{itemize}

However, being computationally intensive,  MPC is more suitable for slowly-varying dynamics; whereas, a quadrotor has fast-changing dynamics. Further, MPC is mostly applied via system discretization, while a  {continuous-time system} representation is more practical  \cite{dhar2021indirect}. To illustrate, it is shown in \cite{CHIKASHA2017AMPC} that the Adaptive MPC requires the Recursive Least Square method for parameter estimation on a discrete state-space model. This requires high computational complexity, compared to the computationally light approach in this thesis.

Then, the question arises whether the existing BLF-based solutions can be applied to quadrotor. Most of the BLF-based controllers have been carried out on fully actuated systems (cf. \cite{KPTee2009BLF, liu2016barrier, sachan2019output}); a few recent works on quadrotors using BLF cannot deal with system uncertainty and external disturbances (cf. \cite{dasgupta2019singularity,kumar2020barrier}). Furthermore, as previously discussed, a quadrotor system is underactuated by nature, which indeed poses a great challenge to impose multiple state-constraints through the BLF method.

In view of the above discussions and to the best of the author's knowledge, a control solution for uncertain quadrotor systems under single or multiple state-constraints is still missing at large.

\section{Contribution of this thesis}

Based on the issues faced as described in the above discussion, this thesis attempts to fill a gap in literature by providing major contributions toward the following directions:

\begin{itemize}

    \item A BLF-based robust controller for a quadrotor system is formulated which obeys state-constraints on position and orientation on all six degrees-of-freedom (DoFs). The controller design imposes user-specific constraints on the position and the orientation, and eventually, on all the controlled states (i.e., position, orientation, linear and angular velocity). Additionally, the controller can tackle parametric uncertainties and external disturbances. Compared to \cite{dasgupta2019singularity}, constraints are imposed on all the available DoFs, thereby making the control problem more practical. 
    \item The closed-loop stability is verified analytically, and the performance of the formulated designs are extensively validated via various realistic simulation scenarios, and via experiments performed on a real quadrotor. The results obtained are compared with the state-of-the-art, to validate the efficacy of the proposed design.

\end{itemize}

\section{Background} 
\label{sec11:Background}

\subsection{Choosing Euler-Lagrange Systems}

Before proceeding towards a controller design, it is essential to model the dynamics of the plant. Since controllable plants are diverse, a challenge is to choose a representation of the dynamics that, to a great degree of accuracy, serves as an appropriate model to formulate a controller from. In other words, choosing a generalized representation of system dynamics for a large class of real-life systems eliminates the concern of having to search for different classes of plants for each designed controller. This is where the Euler-Lagrange (EL) representation comes to play. The EL representation is an extremely popular method for system modeling as it covers an enormous range of real-world system dynamics, including underwater vehicles, robotic manipulators, electrical circuits, aerial vehicles and marine vehicle \cite{spong2008robot, roy2019reduced, ye2020switching, roy2020vanishing, rayguru2021output, shukla2021robust}. This work chooses the EL representation because it appropriately models the dynamics of a quadrotor (discussed later in Chapter 3, Section 3.2). A brief overview of the Euler-Lagrange representation is given below, followed by the generalized EL dynamics that have been followed in the thesis.

The EL equations can be defined as a set of second order Ordinary Differential Equations which are solved by minimizing the integral of a smooth function $\mathcal{L}(t,\mathbf{q}(t), \dot{\mathbf{q}}(t))\in\mathbb{R}$, known as the Lagrangian, where $\mathbf q(t), \dot{\mathbf q} (t)\in\mathbb{R}^n$ are the generalized coordinates and the corresponding velocities in a defined configuration. The EL equations evolve by implementing Hamilton's Principle and following fundamental principles in calculus of variations \cite{cassel_2013}.

Let a scalar, known as the action-functional, be defined as:
\begin{align}
    \Gamma &\triangleq \int_{t_i}^{t_f}\mathcal{L}(t,\mathbf{q}(t), \dot{\mathbf{q}}(t))dt
\end{align}
where $\mathbf{q}(t_i)$ and $\mathbf{q}(t_f)$ are the states at the start-point and the end-point respectively of the path taken by the physical system defined by the Lagrangian. Hamilton's Principle states that the path must have stationary action \cite{cassel_2013}, which ultimately leads to the popularly known form of the Euler Lagrange equations:

\begin{align}
    \frac{\partial \mathcal{L}(t,\mathbf{q}(t), \dot{\mathbf{q}}(t))}{\partial \mathbf{q}(t)} - \frac{d}{dt}\left(\frac{\partial \mathcal{L}(t,\mathbf{q}(t), \dot{\mathbf{q}}(t))}{\partial \mathbf{\dot{q}}(t)}\right) &=0 \label{EL_conservative}
\end{align}

In mechanical systems, the introduction of non-conservative forces like control input, friction, drag, and other external disturbances modifies \eqref{EL_conservative} to the following form, as elaborated in \cite{roy2019adaptive_book}:

\begin{align}
    \frac{d}{dt}\left(\frac{\partial \mathcal{L}}{\partial \mathbf{\dot{q}}}\right)-\frac{\partial \mathcal{L}}{\partial \mathbf{q}}&=\boldsymbol{\tau}+\mathbf{d_s} - \mathbf{h}(\mathbf{q},{\dot{\mathbf q}}) \label{EL_non_conservative}
\end{align}
where $\boldsymbol{\tau} \in \mathbb{R}^n$ is the control input, $\mathbf{d_s} \in \mathbb{R}^n$ consist of the external disturbances, and $\mathbf{h(q,} \dot{\mathbf q}) \in \mathbb{R}^n$ represent frictional forces.

Eventually, for the purpose of this thesis, the following generalized Euler-Lagrange dynamics of $n$ DoFs is considered:
\begin{align}
    \mathbf{M(q)}\Ddot{\mathbf{q}} + \mathbf{C}(\mathbf{ q},\Dot{\mathbf{q}})\Dot{\mathbf{q}} + \mathbf{g(q)} + \mathbf{f}(\Dot{\mathbf{q}}) + \mathbf{d_{s}} &= \boldsymbol{\tau} \label{eq:EL_eqn}
\end{align}

where $\mathbf{M(q)} \in \mathbb{R}^{n \times n}$ the inertia matrix, $\mathbf{C(q}, \Dot{\mathbf{q}})\Dot{\mathbf{q}} \in \mathbb{R}^{n \times n}$ consists of Coriolis and centripetal terms, $\mathbf{g(q)} \in \mathbb{R}^{n}$ consists of gravitational terms, and $\mathbf{f}(\Dot{\mathbf{q}}) \in \mathbb{R}^{n}$ consists of frictional terms. The rest of the symbols have been already defined as they carry the same meanings.

The above equation has some essential properties (cf. \cite{roy2019adaptive_book}) which are heavily exploited in control design:

\begin{itemize}
    \item The generalized inertia matrix, $\mathbf{M(q)}$, is uniformly positive definite
    \item The term, $(\dot{\mathbf{M}}\mathbf{(q)}-2\mathbf{C(q},\dot{\mathbf{q}}))$ is skew-symmetric.
\end{itemize}
Other properties, such as the boundedness of the gravitational term, the disturbances, the inertial and the Coriolis matrices are discussed in detail in Chapters 3 and 4.

\subsection{The Barrier Lyapunov Function Approach}

In state-space analysis of a linear system, the stability is analysed generally via analysing the properties of the state-transition matrix \cite{OgataLinearControl}. However, for highly nonlinear systems where linearization techniques are not appropriate for stability analysis, the Lyapunov Direct Method \cite{khalil2002nonlinear}  is used. It is a powerful tool for stability analysis of a non-linear control system where one chooses a Lyapunov functional and utilizes it to analyse the behavior of the system at the equilibrium.

To illustrate, a generalized non-linear autonomous system is considered:
\begin{align}
    \dot{\mathbf{x}}(t) = \mathbf{f}(\mathbf{x}(t)) \label{eq: non-linear auto sys}
\end{align}
where $\mathbf{x} \in \mathbb{R}^n,~\mathbf{f}:\mathbb{R}^n \xrightarrow{} \mathbb{R}^n$.
Without losing generality, let the equilibrium be at $\mathbf{x}=\mathbf{0}$ (i.e., $\mathbf{f(0)=0}~\forall t\geq t_0,~\text{where}~t_0~\text{is the time at the initial condition of the states})$.

To proceed with the stability analysis, a generalized energy-like scalar functional $V=V(\mathbf{x})$, called the Lyapunov functional, is chosen. According to the Lyapunov stability criterion, the system is stable at $\mathbf{x=0}$ if:
\begin{itemize}
    \item $V(\mathbf{x})$ is positive-definite; i.e., $V>0 ~\forall \mathbf{x}\neq 0~\text{and}~V(\mathbf{0})=0$; and
    \item $\dot{V}= \nabla V(\mathbf{x})^\mathsf{T}\mathbf{f}(\mathbf{x})$ is negative semi-definite. For asymptotic stability, $\dot{V}$ must be negative definite.
\end{itemize}
\begin{figure*}[ht!]
    \includegraphics[width=4in, height=3.5in]{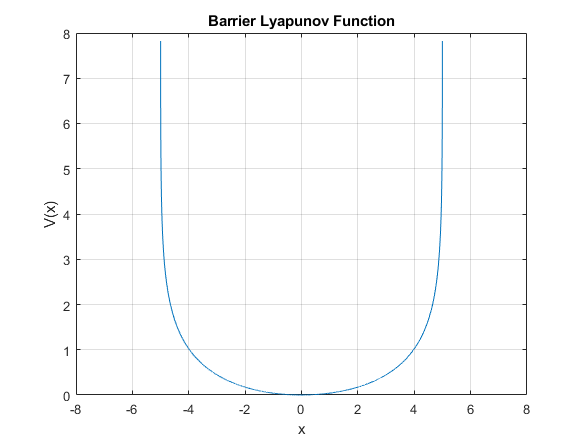}
    \centering
    \caption{A symmetrical Barrier Lyapunov Function pertaining to state $x\in\mathbb{R}$, with $k_{b}=5$}
     \label{fig:BLF_graph}
\end{figure*}
Barrier Lyapunov Functions \cite{KPTee2009BLF} are a special class of Lyapunov functions, where $V(\mathbf{x})$ is defined on the region $\mathcal{D}$ which contains the origin, such that the function is continuous and differentiable at all points in $\mathcal{D}$, and as $V \rightarrow \infty$, $x_i$ approach the boundary of $\mathcal{D}$, provided that $x_i(0) \in \mathcal{D}$; where $i=1,~2,~3,...,~n$.

Following the above definition, the following Barrier Lyapunov function on the system \eqref{eq: non-linear auto sys} with $n=1$ can be constructed:

\begin{align}
    V = \ddfrac{1}{2}\log{\left(\ddfrac{k_{b}^{2}}{k_{b}^{2} - x^{2}}\right)} \label{eq: BLF_1}
\end{align}

where the scalar $k_{b} \in \mathbb{R}^{+}$. It can be inferred from the  above equation that the domain of $V$ is positive definite for $x \in (-k_{b},~k_{b})$. It will be shown in the subsequent chapters that if $|x(0)|<|k_{b}|$, $x$ will never reach $k_{b}$, and $-k_{b}<x(t)<k_{b}~\forall t\geq0$. The nature of the function is illustrated in Fig. \ref{fig:BLF_graph}, which shows a symmetrical Barrier Lyapunov Function (i.e., the constraints are symmetrical about the y-axis), where  $k_{b}=5$. These properties make it suitable for constraint handling control problems, and hence, its properties have been extensively utilized in this thesis.






\section{Thesis Organization} 
\label{Organization}

The thesis is organized as follows:

\begin{itemize}
    \item Chapter 2 elaborates on the implementation and the effectiveness of the Barrier Lyapunov Function on multi-DoF EL dynamics to provide an explanation as to why it is superior to conventional `unconstrained' control schemes which depend on tuning to control overshoot in the transient response. The control design is discussed, followed by the stability analysis and a simulation scenario of a two DoF robotic manipulator.
    
    \item Chapter 3 introduces the robust control design formulated for a quadrotor under spatial constraints. The dynamics of the quadrotor are provided, followed by the control solution. This is followed by the stability analysis. Eventually, two simulation scenarios are provided, and the results are discussed.
    
    \item Chapter 4 extends the work of Chapter 3 by implementing full state-constraints on the quadrotor along with robust control. Similar to Chapter 3, the control solution is followed by the stability analysis. Finally, comparative experimental results are provided, followed by discussion of the results.
    
    \item Chapter 5 concludes the thesis and briefly discusses possible future work. 
\end{itemize}


\chapter{Barrier Lyapunov Function (BLF) based Control Design: an Overview }
\label{ch:chap2}
\section{Introduction}

The imposition of user-specified constraints on the output states impacts the transient performance of the system under control, by specifying the maximum allowable deviation of the constrained states from its desirable quantity. This chapter focuses on the design of a Barrier Lyapunov Function based controller (cf. \cite{KPTee2009BLF}) on an n-DoF Euler-Lagrange system, and validates the effectiveness of the control design through simulation of a 2-DoF robotic manipulator. The results are compared with conventional PID control to emphasize why extensive tuning is required for specified performance limits, whereas the BLF controller provides explicit control parameters that bar the states from violating the imposed constraints. This chapter serves as an overview of the efficacy of constrained trajectory tracking control, and serves as an inspiration for the control designs formulated for a quadrotor in Chapters 3 and 4.

The rest of this chapter is organized in the following manner: Section 2.2 provides the system dynamics for the control design, Section 2.3 provides the control problem, Section 2.4 shows the control solution, Section 2.5 elaborates on the stability analysis for the symmetrical BLF scenario, Section 2.6 elaborates on the stability analysis for asymmetrical BLF scenario, Section 2.7 contains the simulation results and analysis, and finally, Section 2.8 concludes the chapter.

This chapter uses the following notations: $\diag{\{\cdot \cdot \cdot\}}$ and $\mathbf{I}$ represent a diagonal and an identity matrix respectively.

\section{System Dynamics}

This chapter follows the EL dynamics equation \eqref{eq:EL_eqn} defined in Chapter 1.
Rewriting the equation, we get:
\begin{align}
    \mathbf{M(q)}\ddot{\mathbf{q}} + \mathbf{H}(\mathbf{q},\dot{\mathbf{q}}) = \boldsymbol{\tau} \label{eq:EL_eqn_shortened}
\end{align}
where $\mathbf{H}(\mathbf{q},\dot{\mathbf{q}})\triangleq \mathbf{C}(\mathbf{q},\dot{\mathbf{q}})\dot{\mathbf{q}} + \mathbf{g(q)} + \mathbf{f}(\dot{\mathbf{q}}) + \mathbf{d_{s}}$. The symbols and their meanings are already defined in Chapter 1.

Defining the states $\mathbf{x_{1} = q}$, and $\mathbf{x_{2}} = \dot{\mathbf{q}}$, the equation \eqref{eq:EL_eqn_shortened} is rewritten as:
 \begin{align}
 \dot{\mathbf{x}}_{\mathbf{1}}&=\mathbf{x_{2}}\\
 \dot{\mathbf{x}}_\mathbf{2} &= \mathbf{-M}^{-1}\mathbf{(x_{1})H(x_{1},x_{2})} + \mathbf{M}^{-1}\mathbf{(x_{1})}\boldsymbol{\tau} \label{eq: EL_eqn_alt}
 \end{align}
 
\section{Control problem}

Under the assumption that the desired trajectories to be tracked, $\mathbf{q_d}$, are sufficiently smooth and bounded, the control objective is defined as tracking under position constraints. In this chapter, both symmetrical and asymmetrical constraints are imposed on the system. Mathematically, if \\ $\mathbf{z_1}=[z_{11},~z_{12},...~z_{1n}]$ are position tracking errors, the objective is to design a control method that ensures that:
\begin{itemize}
    \item $|z_{1i}(t)|<|k_{b_i}|$ $\forall i$ (where $k_{b_i}$ are user-defined scalars) for the symmetrical constraint scenario, and
    \item $-k_{a_i}<z_{1i}(t)<k_{b_i}$ $\forall i$ (where $k_{a_i}$ and $k_{b_i}$ are user-defined scalars) for the asymmetrical constraint scenario.
\end{itemize}

\section{Control solution}
The following terms are defined:
\begin{align}
\mathbf{z_{1}} &= \mathbf{x_{1}-q_{d}}\\
\mathbf{z_{2}} &= \mathbf{x_{2}-\boldsymbol{\alpha}} \label{eq:trajectory}
\end{align}
where $\mathbf{z_1}$ is the tracking error, $\mathbf{z_2}$ is an auxiliary error variable, and $\boldsymbol{\alpha}$ is an auxiliary control variable defined subsequently.

The control input for the symmetrical constraint scenario is given as:
\begin{subequations}\label{eq:control_input_symmetrical}
\begin{align}
\boldsymbol{\tau}&=\mathbf{Mu_{sym}+H} \\
\boldsymbol{\alpha_1}&= -\mathbf{D_{1}}^{-1}\mathbf{K_{1}z_{1}} + \dot{\mathbf{q}}_\mathbf{d} \\
\mathbf{u_{sym}}& = \dot{\boldsymbol{\alpha}}_\mathbf{1} - \mathbf{K_{2}z_2} - \mathbf{D_{1}z_{1}}
\end{align}
\end{subequations}
where $\mathbf{K_2}~\text{and}~\mathbf{K_2}\in \mathbb{R}^{n \times n}$ are user-defined positive definite matrices and \\
\noindent $\mathbf{D_1} \triangleq \diag\Bigg\{{\ddfrac{1}{k_{b_1}^2-z_{11}^2}, \ddfrac{1}{k_{b_2}^2-z_{12}^2}, \cdot \cdot \cdot, \ddfrac{1}{k_{b_n}^2-z_{1n}^2}}\Bigg\}$.
Similarly, the control input for the asymmetrical constraint scenario is given as:
\begin{subequations}\label{eq:control_input_asymmetrical}
\begin{align}
\boldsymbol{\tau}&=\mathbf{Mu_{asy}+H} \\
\boldsymbol{\alpha_2}&=-\mathbf{D_{2}}^{-1}\mathbf{K_{1}z_{1}} + \Dot{\mathbf{q}}_\mathbf{d} \\
\mathbf{u_{asy}}&= \dot{\boldsymbol{\alpha}}_\mathbf{2} - \mathbf{K_{2}z_2} - \mathbf{D_{2}z_{1}}
\end{align}
\end{subequations}

where  $\mathbf{D_2} \triangleq \diag\Bigg\{{\ddfrac{r_1}{k_{b_1}^2-z_{11}^2}, \ddfrac{r_2}{k_{b_2}^2-z_{12}^2}, \cdot \cdot \cdot, \ddfrac{r_n}{k_{b_n}^2-z_{1n}^2}}\Bigg\}+ \diag\Bigg\{{\ddfrac{1-r_1}{k_{b_1}^2-z_{11}^2}, \ddfrac{1-r_2}{k_{b_2}^2-z_{12}^2}, \cdot \cdot \cdot, \ddfrac{1-r_n}{k_{b_n}^2-z_{1n}^2}}\Bigg\}$.

The value $r_i$ is formulated as:

$r_i = \begin{cases}
1 & z_{1i}>0\\
0 & z_{1i}\leq0
\end{cases}$

In the following stability analysis, it is proven that the trajectories never violate the constraints and they remain bounded $\forall t\geq0$.

\section{Stability analysis for the symmetrical BLF scenario}

\begin{theorem}
Under the initial conditions $|z_{1i}(0)|< k_{b_i}$ and using the  control law \eqref{eq:control_input_symmetrical}, the tracking error trajectories $z_{1i}$ remain bounded by the constraints $k_{b_i}$ as $|z_{1i}|< k_{b_i}$ with $i=1,2,3,...,n$ $\forall t > 0$.
\end{theorem}
\begin{proof}
The BLF candidate for this case is designed as follows:

\begin{align}
    V_{1}&=\sum_{i=1}^{n}\frac{1}{2}\log{\left(\frac{k_{b_i}^2}{k_{b_i}^2- z_{1i}^2}\right)} + \frac{1}{2}\mathbf{z_2}^\mathbf{T}\mathbf{z_2}. \label{eq:sym_blf}
\end{align}
\\
\noindent
From the time-derivative of \eqref{eq:sym_blf} and applying the control input equations in \eqref{eq:control_input_symmetrical}, we get: 
\begin{align}
    \dot{V}_1 &= \mathbf{z_{1}^{T}}\mathbf{D_1}\dot{\mathbf{z}}_\mathbf{1} + \mathbf{z_{2}^{T}}\dot{\mathbf{z}}_\mathbf{2} \nonumber \\
    &= \mathbf{z_{1}^{T}}\mathbf{D_1}(\mathbf{z_{2}}+\boldsymbol{\alpha}-\dot{\mathbf{q}}_\mathbf{d}) + \mathbf{z_{2}^{T}}\Dot{\mathbf{z}}_\mathbf{2} \nonumber \\
    &= -\sum_{i=1}^{2} \mathbf{z_{i}^{T}K_{i}z_{i}}<0.\label{eq:sym_BLF_deriv}
\end{align}

As $\mathbf{K_i}$ is user-defined positive definite matrix, the trajectories remain bounded from (\ref{eq:sym_BLF_deriv}) and consequently, the constraints are never violated.
\end{proof}
\section{Stability analysis for the Asymmetrical BLF scenario}
\begin{theorem}
Under the initial conditions $-k_{a_i}<z_{1i}(0)< k_{b_i}$ and using the  control law \eqref{eq:control_input_asymmetrical}, the tracking error trajectories $z_{1i}$ remain bounded by the constraints $k_{a_i}$ and $k_{b_i}$ as $-k_{a_i}<z_{1i}(t)< k_{b_i}$ for $i=1,2,3,...,n$ $\forall t > 0$.
\end{theorem}
\begin{proof}
In this case, the Lyapunov function candidate is designed as:
\begin{align}
V_{1} = \sum_{i=1}^{n}\frac{1}{2}r_i{\log{\left(\frac{k_{b_i}^2}{k_{b}^2-z_{1i}^2}\right)}} + \sum_{i=1}^{n}\frac{1}{2}(1-r_i)\log{\left(\frac{k_{a_i}^2}{k_{a_i}^2-z_{1i}^2}\right)} + \frac{1}{2}\mathbf{z_2}^{T}\mathbf{z_2}. \label{eq: BLF_asym}
\end{align}

From the time-derivative of \eqref{eq: BLF_asym} and the control input equations \eqref{eq:control_input_asymmetrical}, similar calculations as in the symmetrical constraints case follow, which proves that the trajectories remain bounded and the constraints are not violated.

\end{proof}

\section{Simulation and Results}

\begin{figure}[ht!]%
    \centering
    \includegraphics[width=11.7cm]{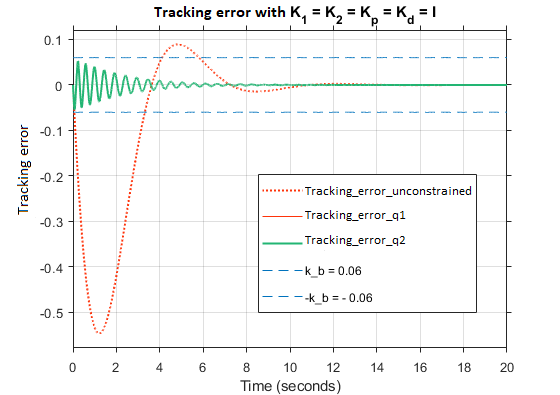}
    \caption{Trajectory tracking error comparison by employing symmetrical constraints with SBLF candidate vs unconstrained trajectory error}%
    \label{fig:SBLF_ABLF1}%
\end{figure}
\begin{figure}[ht!]%
    \centering
    \includegraphics[width=11.7cm]{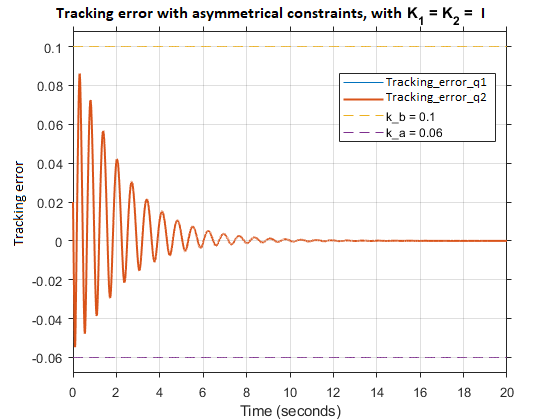}%
    \caption{Trajectory error with asymmetrical constraints}%
    \label{fig:SBLF_ABLF2}%
\end{figure}
\begin{figure}[ht!]%
\centering
\includegraphics[width = 14cm]{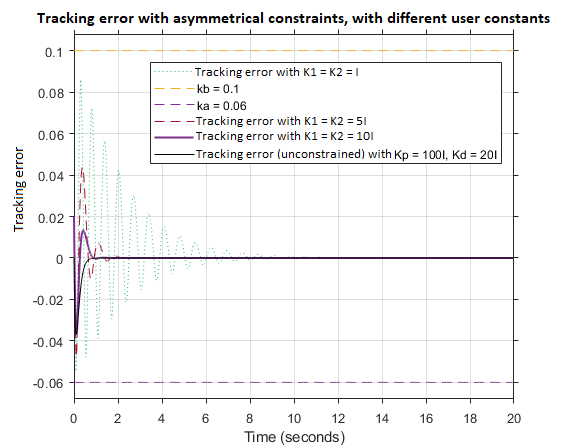}
\caption{Trajectory error with different user defined parameters employed on the asymmetrical constraints case vs unconstrained case.} \label{ABLF_different_values}
\end{figure}

Simulation has been carried out on a 2-DoF system \cite{roy2020new}, having the following parameters that comply with the system dynamics of \eqref{eq:EL_eqn}:\\
\\$
\mathbf{M} =
\begin{bmatrix}
M_{11} & M_{12} \\
M_{21} & M_{22}
\end{bmatrix}
$, and $\mathbf{q}=\begin{bmatrix}
q_{1} \\ q_{2}
\end{bmatrix}$
, where $M_{11} = (m_{1}+m_{2})l_{1}^{2} + m_{2}l_{2}(l{2}+2l_{1}\cos{(q_{2}}))$, $M_{12} =m_{2}l_{2}(l{2}+l_{1}\cos{(q_{2})})$, and $M_{22}=m_{2}l_{2}^2$
\\
\\
$\mathbf{C} =
\begin{bmatrix}
-m{2}l_{1}l_{2}\sin{(q_{2})}\Dot{q_{2}} & -m{2}l_{1}l_{2}\sin{q_{2}}(\Dot{q_{1}}+\Dot{q_{2}}) \\
0 & m{2}l_{1}l_{2}\sin{(q_{2})}\Dot{q_{2}}
\end{bmatrix}
$\\

\noindent
$
\mathbf{g} =
\begin{bmatrix}
m_{1}l_{1}g\cos{(q_{1})} + m_{2}g(l_{2}\cos{(q_{1}+q_{2})} + l_{1}\cos{(q_{1})}) \\
m_{2}gl_{2}\cos{(q_{1}+q_{2})}
\end{bmatrix}
$, 
$\mathbf{f} =
\begin{bmatrix}
f_{v1}\sgn{(\Dot{q_{1}})} \\
f_{v2}\sgn{(\Dot{q_{2}})}
\end{bmatrix}$, 
$\mathbf{d}=
\begin{bmatrix}
0.5\sin{(0.5t)} \\
0.5\sin{(0.5t)}
\end{bmatrix}$,
\\
\\
\\
where the following values are chosen:
$m_{1} =$ 10 kg, $m_{2}=$ 5 kg, $l_{1}$=0.2 m, $l_{2} =$ 0.1 m, $f_{v1} = f_{v2} = $ 0.5, and $g = $ 9.81 m/s\textsuperscript{2}.
For the symmetrical constraint case, the desired trajectory is chosen as $\mathbf{q_{d}} = \left[\sin(t), \sin(t)\right]^\mathbf{T}$, $\mathbf{k_b} = [0.06,0.06]^\mathbf{T}$.
For the asymmetrical constraint case,  $\mathbf{q_{d}} = \left[\\-0.02 + \sin(t),\\-0.02+ \sin(t)\right]^\mathbf{T}$.
The constraints are selected as $\mathbf{k_{b}} = [0.1,0.1]^\mathbf{T}$, and $\mathbf{k_{a}} = [0.06, 0.06]^\mathbf{T}$.

The performance of the BLF controller is compared with the PID controller. Figure \ref{fig:SBLF_ABLF1} shows trajectory tracking error comparison in the symmetrical constraints case. It can be clearly notices that with PID controller (the red dotted line), the error violates the user-specified constraints, whereas the similarly tuned BLF controller (green) never violates them. For comparable maximum deviation with the BLF controller, the PID controller needs to be re-tuned, which becomes a major drawback in uncertain scenarios. Similar performance is observed in Fig. \ref{fig:SBLF_ABLF2} by the BLF controller, where the tracking error obeys asymmetrical constraints. In Fig. \ref{ABLF_different_values}, it is observed that to further reduce oscillations about zero, the gains $\mathbf K_1$ and $\mathbf K_2$ can be tuned, without having any concern for the overshoot as the controlled trajectories obey the bounds regardless.

\section{Conclusion}

In this chapter, a Barrier Lyapunov Function based controller is designed and implemented on a multi-DoF EL system. The performance of the controller is compared with the PID, its `unconstrained' counterpart. The results clearly show the efficacy of the BLF based control design. The next chapter utilizes the Barrier Lyapunov approach to tackle the challenging issues faced by a quadrotor while manoeuvring under spatial constraints with limited knowledge on the system dynamics.


\chapter{Robust Manoeuvring of Quadrotor in Constrained Space}
\label{ch:chap3}
\section{Introduction}

Nowadays, quadrotors are widely researched for their immense versatility. In the first chapter, several challenges faced by a quadrotor in relief operations and ration supply scenarios were discussed, which served as strong motivation for the thesis. In the previous chapter, the superiority of the BLF based controller over the conventional PID controller was successfully demonstrated to emphasize the drawbacks faced in conventional control methods pertaining to tuning and re-tuning to achieve a desired performance limit. However, the challenge of formulating such a controller on a quadrotor is two-fold. Firstly, most of the existing works consider fully actuated system with complete knowledge of the dynamics, whereas, the quadrotor dynamics is underactuated, implying that the desired performance under spatial constraints even with full knowledge of the quadrotor dynamics and under ideal conditions
is difficult. It is shown in \cite{dasgupta2019singularity}, that the quadrotor cannot be commanded to follow any arbitrary trajectory and orientation in space. Secondly, the presence of dynamic uncertainty and external disturbances complicate the difficulty even further. 

To illustrate, a quadrotor passing through a pipeline structure where it cannot deviate more than $0.5$ m from its desired path in vertical and horizontal positions to avoid collision with boundary walls. Now, the presence of disturbances, such as wind gust, could easily perturb the quadrotor beyond the maximum allowed deviation, thereby resulting in operational hazard. Even with conventional robust or adaptive control, as discussed in Chapter 1, it is time-consuming to tune the control gains to achieve guaranteed user-specified accuracy.




In view of the challenges faced in works attempting to address the above issue as described in Chapter 1, and to the best of the author's knowledge, a comprehensive control solution for quadrotor in the presence of space constraints (i.e., under predefined position and attitude accuracy) and under {system uncertainty and external disturbances} is still missing at large. 
Toward this direction, this chapter has majorly contributed in the following ways:

\begin{itemize}
    \item A BLF-based robust controller for quadrotor is formulated, which can tackle uncertain system dynamic parameters, payload and external disturbances, while obeying constraints imposed by the user on position and attitude.
    \item Differently from \cite{dasgupta2019singularity}, these constraints are imposed on all six degrees-of-freedom, i.e., on three position and three attitude angles, to make the control problem more suitable under space constraints. 
    \item The closed-loop stability is analysed analytically, and simulation results compared to the state-of-the-art confirm the effectiveness of the proposed scheme in scenarios such as manoeuvring through pipes and hoops and delivering payloads of different masses.
    
\end{itemize}
It is worth remarking that constraint on tracking accuracy being the main focus, the present work does not consider actuator saturation since it compromises tracking performance as a trade-off \cite{KPTee2009BLF, liu2016barrier,dasgupta2019singularity,
sachan2019output} (cf. Remark 4 for further discussion), and this will be taken up as a future work. 

The rest of the chapter is organised as follows: Section 3.2 describes the dynamics of quadrotor and the control problem; Section 3.3 details the proposed control framework, while corresponding stability analysis is provided in Section 3.4; comparative simulation results are provided in Section 3.5, while Section 3.6 provides concluding remarks.

The following notations are commonly used in this and subsequent chapters: $|| \cdot ||$ and $\lambda_{\min}(\cdot)$ are 2-norm and minimum eigenvalue of $(\cdot)$; $\sat(\mathbf{x}, k)$ denotes a saturation function defined as $\sat(\mathbf{x}, k)= \mathbf{x}/|| \mathbf{x} ||$ if $||\mathbf{x}|| \geq k$ and $\sat(\mathbf{x}, k)= \mathbf{x}/k$ if $||\mathbf{x}|| < k$; $\diag{\{\cdot \cdot \cdot\}}$ and $\mathbf{I}$ represent a diagonal and an identity matrix respectively. 

\section{Quadrotor System Dynamics and Problem Formulation}

\begin{figure}[t]
    \includegraphics[width=4in, height=3.5in]{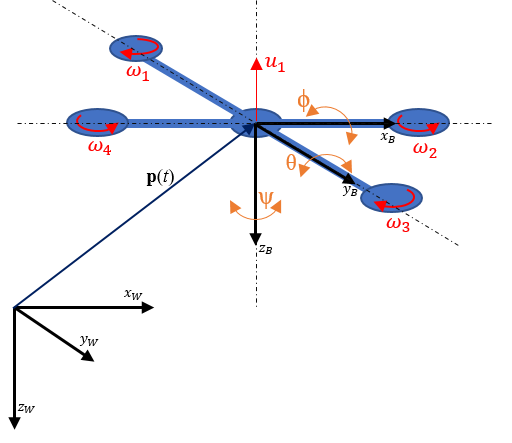}
    \centering
    \caption{The quadrotor model showing the quadrotor frame (\textit{B}); the world frame (\textit{W}); the position vector $\mathbf{p}(t)$; roll($\phi$), pitch($\theta$) and yaw($\psi$) angles; the angular velocities of the rotors ($\omega_i$); and the control input $u_1$}.
     \label{fig:quadrotor_dynamics}
\end{figure}

The dynamic model of a quadrotor system via Euler-Lagrange dynamics can be written as follows \cite{Bialy2013RobustAdaptive_UAV}:
\begin{align}
    m\Ddot{\mathbf{p}}+\mathbf{g+d_{p}} &= \mathbf{\tau_{p}} \label{eq:tau_p_chap_3} \\
    \mathbf{J(q)}\Ddot{\mathbf{q}}+\mathbf{C(q,}\Dot{\mathbf{q}})\Dot{\mathbf{q}}+\mathbf{d_{q}} &= \mathbf{\tau_{q}} \label{eq:tau_q_chap_3}\\
    \mathbf{\tau_p} &= \mathbf{R}_B^{W}\mathbf U \label{comb_chap_3}
\end{align}

\noindent where $m$ is the total mass of the system; $\mathbf{p}(t) \triangleq  \begin{bmatrix}
    x(t) & y(t) & z(t)
\end{bmatrix}^T \in \mathbb{R}^{3}$ is the position of the centre of mass of the quadrotor in the Earth-fixed frame; $\mathbf{q}(t) \triangleq  \begin{bmatrix}
    \phi(t) & \theta(t) & \psi(t)
\end{bmatrix}^T \in \mathbb{R}^{3}$ is the attitude vector consisting of the roll ($\phi$), pitch ($\theta$) and yaw ($\psi$) angles; $\mathbf{g} \triangleq  \begin{bmatrix}
    0 & 0 & mg
\end{bmatrix}^T \in \mathbb{R}^{3}$, where $g$ is the acceleration due to gravity in the z-direction; $\mathbf{J(q)} \in \mathbb{R}^{3 \times 3}$ is the inertia matrix; $\mathbf{C(q,\dot{q})} \in \mathbb{R}^{3 \times 3}$ is the Coriolis matrix and the vectors $\mathbf{d_p}, \mathbf{d_q} \in \mathbb{R}^{3}$ represent effect of external disturbances (e.g., wind, gust); $\boldsymbol \tau_{\mathbf q} \triangleq 
	\begin{bmatrix}
	u_2(t) & u_3(t) & u_4(t)
	\end{bmatrix}^T\in\mathbb{R}^3$ denotes the control inputs for roll, pitch and yaw; $\boldsymbol \tau_{\mathbf p}(t) \in \mathbb{R}^3$ is the generalized control input for position tracking in Earth-fixed frame, with $\mathbf U(t)\triangleq
	\begin{bmatrix}
	0 & 0 & u_1(t)
	\end{bmatrix}^T\in \mathbb{R}^3$ being the force vector in body-fixed frame and $\mathbf R_B^W \in\mathbb{R}^{3\times3}$ being the $Z-Y-X$ Euler angle rotation matrix describing the rotation from the body-fixed coordinate frame to the Earth-fixed frame (see Fig. \ref{fig:quadrotor_dynamics}), given by
	\begin{align}
	\mathbf R_B^W =
	\begin{bmatrix}
	c_{\psi}c_{\theta} & c_{\psi}s_{\theta}s_{\phi} - s_{\psi}c_{\phi} & c_{\psi}s_{\theta}c_{\phi} + s_{\psi}s_{\phi} \\
	s_{\psi}c_{\theta} & s_{\psi}s_{\theta}s_{\phi} + c_{\psi}c_{\phi} & s_{\psi}s_{\theta}c_{\phi} - c_{\psi}s_{\phi} \\
	-s_{\theta} & s_{\phi}c_{\theta} & c_{\theta}c_{\phi}
	\end{bmatrix}, \label{rot_matrix}
	\end{align}
where $c_{(\cdot)} , s_{(\cdot)}$ and denote $\cos{(\cdot)} , \sin{(\cdot)}$ respectively.

\subsection{Standard assumptions for the control design}

Before proceeding towards defining the problem and proceeding towards the controller design, the following assumptions highlight the amount of uncertainties in the system dynamics (\ref{eq:tau_p_chap_3})-(\ref{eq:tau_q_chap_3}), the choice of nominal values for the robust control design, and the nature of the trajectories to be tracked:

\begin{assum}[Uncertainties] The terms $m, \mathbf{J}$ and $\mathbf{C}$ can be segregated as $m=\bar{m}+\Delta m,~ \mathbf{J= \bar{J}}+ \Delta \mathbf J$ and $\mathbf{C}=\bar{\mathbf C}+\Delta \mathbf C$ where the segments $\bar{(\cdot)}$ and $\Delta(\cdot)$ represent known nominal part and uncertain part of the dynamics, respectively. $\Delta(\cdot)$ and external disturbances $\mathbf{d_{p}}$ and $\mathbf{d_q}$ are not known instantaneously, but their maximum possible variations are known. 
\end{assum}
\begin{assum}[Choice of nominal values]
The nominal values $\bar{m}$ and $\bar{\mathbf{J}}$ satisfy the following relations
\begin{align}
|(m^{-1}\bar{m}-1)| \leq E_p <1,~ ||\mathbf{J}^{-1}(\mathbf{q})\bar{\mathbf J}(\mathbf{q})-\mathbf{I}|| \leq E_q <1. \label{mass}
\end{align}
\end{assum}

\begin{remark}[Validity of the assumptions]
Assumptions 1-2 are discussed in Chapter 1, and are quite standard in the literature of robotics \cite{roy2019overcoming,roy2020new, spong2008robot, roy2015robust, roy2013robust}. Specifically, Assumption 2 implies that the perturbations cannot be more than $100\%$ of the respective nominal values, thus giving a guideline to select the nominal values for some allowable range of perturbation. In practice, information on the maximum allowable payload/ overall system mass $m$ of a quadrotor is always available. With this, one can design $\bar{m}$ from (\ref{mass}) and consequently design $\bar{\mathbf{J}}$ (to satisfy (\ref{mass})) and $\bar{\mathbf{C}}$ from their structures by following \cite{tang2015mixed}.
\end{remark}

\begin{remark}[Position and attitude tracking co-design]\label{rem_co}
The a priori boundedness assumption of the non-actuated $(x,y)$ dynamics stemming from the reduced-order model (a.k.a collocated) based design (cf. \cite{nicol2011robust, dydek2012adaptive, sankaranarayanan2020introducing, roy2020towards, roy2021artificial}) results in a conservative approach towards the control design, as part of the challenge in tackling an underactuated system. Therefore, in this chapter, the position and attitude tracking co-design method (cf. \cite{bialy2013lyapunov, mellinger2011minimum, sankaranarayanan2020aerial}) has been followed: this allows to design tracking controller for all six degrees-of-freedom (DoF) in dynamics (\ref{eq:tau_p_chap_3})-(\ref{eq:tau_q_chap_3}), rather than tracking only the actuated DoFs $z$ and $(\phi, \theta ,\psi)$ as in the collocated design approach. Note that such design is not decoupled, rather simultaneous as the dynamics (\ref{eq:tau_p_chap_3}) and (\ref{eq:tau_q_chap_3}) are connected via (\ref{comb_chap_3}). 
\end{remark}

The following standard assumption is considered:
\begin{assum}[\cite{bialy2013lyapunov, mellinger2011minimum, tang2015mixed}]\label{assum_des}
Let $\mathbf{p_d} (t) \triangleq
	\begin{bmatrix}
	x_d (t) & y_d (t) & z_d (t)
	\end{bmatrix}^T$ and $\psi _d(t) $ be the desired position and yaw trajectories to be tracked, which are designed to be sufficiently smooth and bounded. 
\end{assum}
\begin{remark}[Desired roll and pitch]
In position and attitude tracking co-design, the desired roll ($\phi_d$) and pitch $(\theta_d$) trajectories are derived based on the computed position control input $\tau_p$ and the desired yaw trajectory $\psi_d$ (cf. \cite{mellinger2011minimum}) and it is discussed in detail in Section 3.3.2.
\end{remark}

\subsection{Control Problem}

Manoeuvring under space constraints implies that the quadrotor cannot violate a specified performance limit, which implies that the desired control objective becomes \textit{tracking under predefined accuracy}. Let $\mathbf{z_{p}}=\lbrace z_{p1},~z_{p2},~z_{p3} \rbrace$ and $\mathbf{z_{q}}=\lbrace z_{q1},~z_{q2},~z_{q3} \rbrace$ be the tracking errors in position $(x,y,z)$ and attitude $(\phi
, \theta, \psi)$ respectively and, $k_{pi}, k_{qi}$ $i=1,2,3$ be six user-defined scalars. Then, the control problem is defined as follows:

Under Assumptions 1-2, to design a robust controller to track a desired trajectory (cf. Assumption 3) such that the tracking accuracy remains within a user-defined region as $|z_{pi}|< k_{pi}, |z_{qi}|< k_{qi}$ with $i=1,2,3$ and $|z_{pi}(0)|< k_{pi}, |z_{qi}(0)|< k_{qi}$ for all time $t > 0$.

Initial conditions lying within the constraints ($|z_{pi}(0)|< k_{pi}, |z_{qi}(0)|< k_{qi}$) is standard in any control literature handling constraints, otherwise the control problem becomes ill-posed and infeasible (cf. \cite{KPTee2009BLF, liu2016barrier,dasgupta2019singularity,sachan2019output}). The following section provides a solution to this control problem.

\section{Proposed Control Solution}

\begin{figure*}
    \includegraphics[width=6in, height=2.6in]{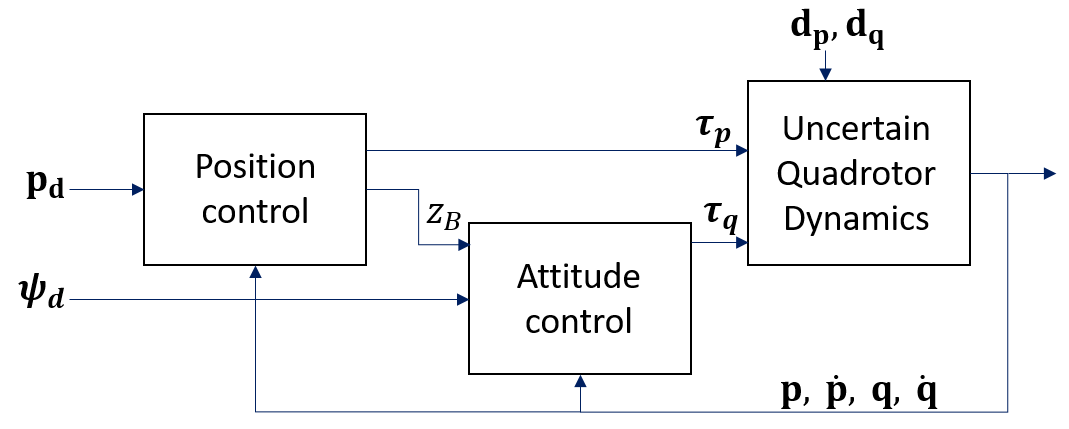}
    \centering
    \caption{Block diagram of the control system}
     \label{fig:BD_Controller}
\end{figure*}

The co-design approach relies on simultaneous design of an outer loop controller for position dynamics (\ref{eq:tau_p_chap_3}) and of an inner loop controller for attitude dynamics (\ref{eq:tau_q_chap_3}). Figure \ref{fig:BD_Controller} shows the overall control framework, while the detailed design is elaborated in the following subsections.

\subsection{Outer Loop Controller Design}
The position tracking error  $\mathbf{z_{p}}=\lbrace z_{p1},~z_{p2},~z_{p3} \rbrace$ and an auxiliary error variable $\mathbf{z_{2p}}$ are defined as
\begin{align}
    \mathbf{z_{p}} &= \mathbf{p-p_{d}}, \label{eq:z_p_chap_3} \\
    \mathbf{z_{2p}} &= \Dot{\mathbf{p}}-\mathbf{\alpha_{p}}, \label{eq:z_2p_chap_3}
\end{align}
where $\mathbf{\alpha_{p}}$ is an auxiliary control variable. The position control law is designed as
\begin{subequations}\label{ctr_p_chap_3}
\begin{align}
    \mathbf{\tau_{p}}&=\Bar{m}\mathbf{\nu_{p}}+\bar{\mathbf{g}},~~\mathbf{\nu_{p}} = \bar{\mathbf{\nu}}_{\mathbf p}+\Delta\mathbf{\nu_{p}}, \label{eq:tau_p_chap_3_expression} \\
    \mathbf{\alpha_{p}} &=-\Lambda_{1p}\mathbf{z_{p}}+\dot{\mathbf p}_{\mathbf d}, \label{alpha_p_chap_3}\\
    \bar{\mathbf{\nu}}_{\mathbf p} &= \dot{\mathbf \alpha}_{\mathbf p}-\mathbf{D_{p}}\mathbf{z_{p}}-\Lambda_{2p}\mathbf{z_{2p}},\label{eq:nu_p_chap_3 expression} \\
    \Delta \mathbf{\nu_{p}}&=-\rho_{p}\sat(\mathbf{z_{2p}}, \epsilon_p), \label{eq:delta_nu_p_chap_3}
\end{align}
\end{subequations}
where $\bar{\mathbf{g}}=[0~0~\bar{m}g]^T$; $\Lambda_{1p} = \diag{\{\gamma_{p_{1}}, \gamma_{p_{2}},\gamma_{p_{3}}\}}$, where $\gamma_{p_{i}}>0$ are user-defined constants; $\mathbf{D_{p}}=\diag{\bigg\{\frac{1}{k_{p_{1}}^{2}-z_{p_{1}}^{2}}, \ddfrac{1}{k_{p_{2}}^{2}-z_{p_{2}}^{2}}, \ddfrac{1}{k_{p_{3}}^{2}-z_{p_{3}}^{2}}\bigg\}}$; $\Lambda_{2p} $ is a user-defined positive definite matrix; $\Delta\mathbf{\nu_{p}}$ is the robust control input, which tackles the uncertainties via the gain $\rho_{p}$ to be defined later and $\epsilon_p>0$ is a gain used to avoid chattering.

%
%
%
Using (\ref{eq:tau_p_chap_3}) and \eqref{eq:tau_p_chap_3_expression}, the time derivative of \eqref{eq:z_2p_chap_3} yields
\begin{align}
    \Dot{\mathbf z}_{2\mathbf p} &= \Ddot{\mathbf{p}}-\Dot{\mathbf \alpha}_{\mathbf p} = m^{-1}(\mathbf{\tau_{p}}-\mathbf{g}-\mathbf{d_{p}})-\Dot{\mathbf \alpha}_{\mathbf p} \nonumber \\ 
    &= \mathbf{\nu_{p}}+E_{p}\mathbf{\nu_{p}}+m^{-1}(\bar{\mathbf{g}}-\mathbf{g}-\mathbf{d_p})-\Dot{\mathbf \alpha}_{\mathbf p} \nonumber \\ 
    &= \Bar{\nu}_{\mathbf p}+\Delta \mathbf{\nu_{p}}+\mathbf{\eta_{p}}-\Dot{\mathbf \alpha}_{\mathbf p} \label{ddot_z2p_chap_3}\\
\text{where} ~~~ \mathbf{\eta_{p}} &\triangleq (m^{-1}\bar{m}-1)\mathbf{\nu_{p}}+m^{-1}(\bar{\mathbf{g}}-\mathbf{g}-\mathbf{d_p}) \label{un_p_chap_3}
\end{align}  
is the \textit{overall uncertainty} in the position dynamics. Having specified the uncertainty structure, $\rho_p$ is to be designed such that $\rho_p \geq ||\mathbf{\eta_{p}}|| $. Then, using relation (\ref{mass}), from (\ref{un_p_chap_3}) we have 
\begin{align}
    \rho_{p} &\geq E_{p} \|\Bar{\mathbf{\nu}}_{\mathbf p}\|+ E_{p} \rho_{p} +m^{-1} ||(\bar{\mathbf{g}}-\mathbf{g}-\mathbf{d_p}) || \nonumber\\
   \Rightarrow \rho_{p} &\geq \ddfrac{1}{1-E_{p}}(E_{p}\|\Bar{\mathbf{\nu}}_{\mathbf p}\|+m^{-1}(||\bar{\mathbf{g}}-\mathbf{g}||+ ||\mathbf{d_p}||). \label{eq:rho_p_chap_3}
\end{align}
Eventually, $\mathbf U$ is applied via relation (\ref{comb_chap_3}).
\subsection{Inner Loop Controller Design}
To design inner loop controller, the desired roll ($\phi_d$) and pitch ($\theta_d$) angles are to be generated: first, an intermediate coordinate frame is defined as (cf. \cite{mellinger2011minimum}):
\begin{subequations}\label{int_co_chap_3}
\begin{align}
    z_B &= \frac{\mathbf{\tau_{p}}}{||\mathbf{\tau_{p}}||} ,~~y_A = \begin{bmatrix}
    -s_{\psi_d} & c_{\psi_d} & 0
\end{bmatrix}^T \\
    x_B &= \frac{y_A \times z_B}{||y_A \times z_B||} ,~~y_B = z_B \times x_B
\end{align}
\end{subequations}
where $y_A$ is the $y$-axis of the intermediate coordinate frame $A$, $x_B$, $y_B$ and $z_B$ are the desired $x$-axis, $y$-axis and $z$-axis of the body-fixed coordinate frame. The desired yaw angle $\psi_d (t)$ computes the axis $y_A$ and the the computed intermediate axes in (\ref{int_co_chap_3}) finally determine $\phi_d (t)$ and $\theta_d (t)$ \cite{mellinger2011minimum}. These desired orientation angles finally completely define $\mathbf{R_d}$ is the rotation matrix as in (\ref{rot_matrix}) evaluated at ($\phi_d, \theta_d, \psi_d$).

Further, the attitude error can be defined as \cite{mellinger2011minimum}:
\begin{align}
     \mathbf{z_{q}}= \lbrace z_{q1}, z_{q2}, z_{q3} \rbrace &= {\mathbf{((R_d)^T R_B^W - (R_B^W)^T R_d)}}^{v} \label{eq:z_q_chap_3}
\end{align}
where $(.)^v$ represents \textit{vee} map, which converts elements of $SO(3)$ to $\in{\mathbb{R}^3}$ \cite{mellinger2011minimum}. The auxiliary error variable $\mathbf{z_{2q}}$ is defined as:
\begin{align}
    \mathbf{z_{2q}} &= \dot{\mathbf{q}}-\mathbf{\alpha_{q}} \label{eq:z_2q_chap_3}
\end{align}
with $\mathbf{\alpha_{q}}$ being an auxiliary control variable defined subsequently.
The inner loop control law is designed as
\begin{subequations}\label{ctr_q_chap_3}
\begin{align}
    \mathbf{\tau_{q}}&=\Bar{\mathbf{J}}\mathbf{\nu_{q}}+\bar{\mathbf{C}}\dot{\mathbf{q}},~\mathbf{\nu_{q}} = \Bar{\mathbf{\nu}}_{\mathbf q}+\Delta\mathbf{\nu_{q}}, \label{eq:tau_q_chap_3_expression}\\
    \mathbf{\alpha_{q}} &=-\Lambda_{1q}\mathbf{z_{q}}+\dot{\mathbf{q}}_\mathbf{d}, \label{alpha_q_chap_3}\\
    \bar{\mathbf{\nu}}_{\mathbf q} &= \mathbf{\Dot{\alpha}}_{\mathbf q}-\mathbf{D_{q}}\mathbf{z_{q}}-\Lambda_{2q}\mathbf{z_{2q}},\label{nu_q expression} \\
    \Delta \mathbf{\nu_{q}}&=-\rho_{q}\sat(\mathbf{z_{2q}}, \epsilon_q), \label{eq:delta_nu_q_chap_3}
\end{align}
\end{subequations}
where $\mathbf{D_{q}}=\diag{\bigg\{\ddfrac{1}{k_{q_{1}}^{2}-z_{q_{1}}^{2}}, \ddfrac{1}{k_{q_{2}}^{2}-z_{q_{2}}^{2}}, \ddfrac{1}{k_{q_{3}}^{2}-z_{q_{3}}^{2}}\bigg\}}$; $\Lambda_{1q} = \diag{\{\gamma_{q_{1}}, \gamma_{q_2},\gamma_{q_3}\}}$, with $\gamma_{q_i}>0$ being user-defined constants; $\Lambda_{2q}$ is a user-defined positive definite matrix; $\Delta\mathbf{\nu_{q}}$ is the robust control input, which tackles the uncertainties in the attitude dynamics via the gain $\rho_{q}$ to be defined later and $\epsilon_q>0$ is a gain used to avoid chattering. 

Following similar lines to derive (\ref{ddot_z2p_chap_3}) for the outer loop controller, the following is achieved using (\ref{eq:tau_q_chap_3}), \eqref{eq:z_2q} and \eqref{eq:tau_q_chap_3_expression}
\begin{align}
    \Dot{\mathbf z}_{2\mathbf q} &= \Bar{\nu}_{\mathbf q}+\Delta \mathbf{\nu_{q}}+\mathbf{\eta_{q}}-\Dot{\mathbf \alpha}_{\mathbf q}\label{eq:ddot_z2q_chap_3}\\
\text{where} ~~~ \mathbf{\eta_{q}} &\triangleq (\mathbf{J}^{-1}\bar{\mathbf J}-\mathbf{I})\mathbf{\nu_{q}}+\mathbf{J}^{-1}(-\mathbf{d_q}-\Delta \mathbf{C}\dot{\mathbf{q}}) \label{uncer_att_chap_3}
\end{align}
represents the \textit{overall uncertainty} in the attitude dynamics. Using the above expression and relation (\ref{mass}), the robust control gain is designed to be $\rho_q \geq || \mathbf{\eta_{q}} || $, i.e.,
\begin{align}
\rho_{q} &\geq E_q || \Bar{\mathbf{\nu}}_{\mathbf q} || + E_q \rho_q +  || \mathbf{J}^{-1}(-\mathbf{d_q}-\Delta \mathbf{C}\dot{\mathbf{q}})|| \nonumber \\
    \rho_{q} &\geq \ddfrac{1}{1-E_q}(E_q || \Bar{\mathbf{\nu}}_{\mathbf q} || + || \mathbf{J}^{-1} || ( ||\Delta \mathbf{C}||||\dot{\mathbf{q}}|| + || \mathbf{d_q}||   ).  \label{eq:rho_q_chap_3}
\end{align}
\begin{remark}[Choice of gains and feasibility]
It may appear from (\ref{eq:nu_p expression}) and (\ref{nu_q expression}) that the gains $\mathbf{D_p,D_q}$ will become infeasible if ${z_{pi}}=k_{pi}$ and ${z}_{qi}=k_{qi}$ for any $i=1,2,3$ and $t>0$. However, the subsequent closed-loop stability analysis will show that ${z_{pi}}<k_{pi}$ and ${z}_{qi}<k_{qi}$ $\forall t>0$ and infeasibility is avoided under the proposed design. Nevertheless, very small $k_{pi},k_{qi}$ will render better accuracy albeit with higher control gains and control input demand. Therefore, these choices should be made according to application requirements. 
\end{remark}

\section{Stability Analysis}
\begin{theorem}
Under Assumptions 1-3 and initial conditions $|z_{pi}(0)|< k_{pi}, |z_{qi}(0)|< k_{qi}$ and using the proposed robust control laws (\ref{ctr_p_chap_3}) and (\ref{ctr_q_chap_3}), the tracking error trajectories $z_{pi},z_{qi}$ remain bounded by the constraints $k_{pi}, k_{qi}$ as $|z_{pi}|< k_{pi}, |z_{qi}|< k_{qi}$ with $i=1,2,3$ $\forall t > 0$.
\end{theorem}
\begin{proof}
Closed-loop stability analysis is carried via the following barrier Lyapunov function (cf. \cite{KPTee2009BLF} for definition and properties) candidate:
\begin{align}
    V&=V_p + V_q \label{eq:lyap_chap_3} \\
\text{where}~~    V_p &= \ddfrac{1}{2}\sum_{i=1}^{3}\log{\left(\ddfrac{k_{p_{i}}^2}{k_{p_{i}}^{2}-z_{p_i}^{2}}\right)}+\ddfrac{1}{2}\mathbf{z}^T_{\mathbf{2p}}\mathbf{z_{2p}}\nonumber  \\
    V_q &= \ddfrac{1}{2}\sum_{i=1}^{3}\log{\left(\ddfrac{k_{q_{i}}^2}{k_{q_{i}}^{2}-z_{q_i}^{2}}\right)}+\frac{1}{2}\mathbf{z}^{T}_{\mathbf{2q}}\mathbf{z_{2q}}  \nonumber
\end{align}

\noindent Using (\ref{eq:z_2p_chap_3})-\eqref{ddot_z2p_chap_3} and (\ref{eq:z_2q_chap_3})-\eqref{eq:ddot_z2q_chap_3}, the time derivative of \eqref{eq:lyap_chap_3} yields
\begin{align}
    \dot{V}_p &= \sum_{i=1}^{3}\ddfrac{z_{p_i}^{T}\dot{z}_{p_i}}{k_{p_i}^{2}-z_{p_i}^{2}} + \mathbf{z}^{T}_{\mathbf{2p}}\dot{\mathbf z}_{\mathbf{2p}} = \mathbf{z}^{T}_{\mathbf p} \mathbf{D_p}\dot{\mathbf z}_{\mathbf p} + \mathbf{z}^{T}_{\mathbf{2p}}\dot{\mathbf z}_{\mathbf{2p}} \nonumber\\
    &= \mathbf{z}^{T}_{\mathbf p} \mathbf{D_p(z_{2p}}+\alpha_{\mathbf p}-\dot{\mathbf p}_{\mathbf d}) + \mathbf{z}^{T}_{\mathbf{2p}}(\bar{\nu}_{\mathbf p}+\Delta \nu_{\mathbf p}+\eta_{\mathbf p}-\dot{\alpha}_{\mathbf p}) \nonumber\\
    &= -\mathbf{z}^{T}_{\mathbf p} \Lambda_{1p}\mathbf{D_p {z}_p}-\mathbf{z}^{T}_{\mathbf{2p}}\Lambda_{2p}\mathbf{{z}_{2p}}+\mathbf{z}^{T}_{\mathbf{2p}}(\Delta\mathbf{{\mathbf{\nu_p}+\eta_p})}, \label{eq:lyap_chap_3_de_p} \\
    \dot{V}_q &= \sum_{i=1}^{3}\frac{z_{q_i}^{T}\dot{z}_{q_i}}{k_{q_i}^{2}-z_{q_i}^{2}} + \mathbf{z}^T_{\mathbf{2q}} \dot{\mathbf z}_{\mathbf{2q}}= \mathbf{z}^{T}_{\mathbf{q}}\mathbf{D_q}\dot{\mathbf z}_{\mathbf q} + \mathbf{z}^{T}_{\mathbf{2q}} \dot{\mathbf z}_{\mathbf {2q}} \nonumber\\
    &= \mathbf{z}^{T}_{\mathbf{q}}\mathbf{D_q}(\mathbf{z_{2q}}+\mathbf{\alpha_q}-\dot{\mathbf{q}}_\mathbf{d}) + \mathbf{z}^{T}_{\mathbf{2q}}(\bar{\nu}_{\mathbf q}+\Delta\mathbf{\nu_q+\mathbf{\eta_q}}-\dot{\alpha}_{\mathbf q}) \nonumber\\
    &= -\mathbf{z}^{T}_{\mathbf{q}}\Lambda_{1q}\mathbf{D_q}{\mathbf z}_{\mathbf q}-\mathbf{z}^{T}_{\mathbf{2q}}\Lambda_{2q}{\mathbf z}_{\mathbf{2q}}+\mathbf{z}^{T}_{\mathbf{2q}}(\Delta\mathbf{{\mathbf{\nu_q}+\eta_q})}. \label{eq:lyap_chap_3_de_q}
\end{align}
Further simplifications from (\ref{eq:lyap_chap_3_de_p})-(\ref{eq:lyap_chap_3_de_q}) yield
\begin{align}
    \dot{V}\leq & -\sum_{i=1}^{3}\gamma_{p_i}\log{\left(\ddfrac{k_{p_i}^{2}}{k_{p_i}^2-z_{p_i}^2}\right)} -\gamma_{q_i}\log{\left(\ddfrac{k_{q_i}^{2}}{k_{q_i}^2-z_{q_i}^2}\right)} \nonumber \\
    &- \lambda_{\min}(\Lambda_{2p})\|\mathbf{z_{2p}}\|^{2} - \lambda_{\min}(\Lambda_{2q})\|\mathbf{z_{2q}}\|^{2} \nonumber \\
    & +\mathbf{z}^{T}_{\mathbf{2p}}(\Delta\mathbf{\nu_p+\eta_p})
    +\mathbf{z}^{T}_{\mathbf{2q}}(\Delta\mathbf{\nu_q+\eta_q}). \label{eq:lyap_chap_3_der_inequality}
\end{align}
The first two terms in the aforementioned inequality are derived from the fact that $\log{\left(\ddfrac{k_x^{2}}{k_x^{2}-x^{2}}\right)} <   \ddfrac{x^{2}}{k_x^{2}-x^{2}}$ within any compact set $\Omega: |x(t)|<k_x$ $\forall t \geq 0$ and any $k_x \in \mathbb{R}^{+}$ \cite{Liu2016Robust_BLF}.

The definition of $V$ yields
\begin{align}
    V &\leq \sum_{i=1}^{3}\log{\left(\ddfrac{k_{p_i}^{2}}{k_{p_i}^{2}-z_{p_i}^{2}}\right)} + \log{\left(\ddfrac{k_{q_i}^{2}}{k_{q_i}^{2}-z_{q_i}^{2}}\right)} +\|\mathbf{z_{2p}}\|^{2} + \|\mathbf{z_{2q}}\|^{2}. \label{eq:lyap_chap_3_ineq}
\end{align}
Using \eqref{eq:lyap_chap_3_ineq}, from \eqref{eq:lyap_chap_3_der_inequality} we have
\begin{align}
    \dot{V} \leq - \varsigma V + \mathbf{z}^{T}_{\mathbf{2p}} (\Delta \mathbf{\nu_p+\eta_p})  + \mathbf{z}^{T}_{\mathbf{2q}}(\Delta \mathbf{\nu_q+\eta_q}) \label{eq:lyap_chap_3_der_varsigma}
\end{align}
where $\varsigma \triangleq \min \lbrace \min\lbrace \gamma_{p_i} \rbrace, \min\lbrace \gamma_{q_i} \rbrace,\lambda_{\min}(\Lambda_{2p}), \lambda_{\min}(\Lambda_{2q}) \rbrace$.
\\
\noindent Based on the structures of $\Delta \mathbf{\nu_{p}}$ and $\Delta \mathbf{\nu_{q}}$ as in (\ref{eq:delta_nu_p_chap_3}) and (\ref{eq:delta_nu_q_chap_3}), the following four possible cases can be identified.
\\
\noindent \textbf{Case (i)} $\|\mathbf{z_{2p}}\|\geq \epsilon_p$ and $\|\mathbf{z_{2q}}\|\geq \epsilon_q$:
\\
Since $\rho_p\geq\|\mathbf{\eta_p}\|$ and $\rho_q\geq\|\mathbf{\eta_q}\|$ by design, \eqref{eq:lyap_chap_3_der_varsigma} yields
\begin{align}
    \dot{V} &\leq - \varsigma V - \rho_p\|\mathbf{z_{2p}}\|  + \rho_p\|\mathbf{z_{2p}}\| - \rho_q\|\mathbf{z_{2q}}\| + \rho_q\|\mathbf{z_{2q}}\| \nonumber \\
    \implies \dot{V} &\leq -\varsigma V \label{eq:case1_chap_3}
\end{align}
\noindent \textbf{Case (ii)} $\|\mathbf{z_{2p}}\|\geq \epsilon_p$, $\|\mathbf{z_{2q}}\| < \epsilon_q$:
\begin{align}
    \dot{V} &\leq - \varsigma V + \|\mathbf{\eta_q}\|\|\mathbf{z_{2q}}\|
    \leq - \varsigma V + \|\mathbf{\eta_q}\|\epsilon_q
\end{align}
\noindent \textbf{Case (iii)} $\|\mathbf{z_{2p}}\| < \epsilon_p$, $\|\mathbf{z_{2q}}\| \geq \epsilon_q$:
\begin{align}
    \dot{V} &\leq - \varsigma V + \|\mathbf{\eta_p}\|\|\mathbf{z_{2p}}\| 
    \leq - \varsigma V + \|\mathbf{\eta_p}\|\epsilon_p
\end{align}
\noindent \textbf{Case (iv)} $\|\mathbf{z_{2p}}\| < \epsilon_p$, $\|\mathbf{z_{2q}}\| < \epsilon_q$:
\begin{align}
    \dot{V} & \leq - \varsigma V + \|\mathbf{\eta_m}\|\epsilon_m
\end{align}
where $|| \mathbf{\eta_m}||\triangleq \max \lbrace ||\mathbf{\eta_p}||, || \mathbf{\eta_q} || \rbrace$, and $\epsilon_m \triangleq \max \lbrace \epsilon_p, \epsilon_q \rbrace$. Replacing (\ref{alpha_p_chap_3}) into (\ref{eq:z_2p_chap_3}) and (\ref{alpha_q_chap_3}) into (\ref{eq:z_2q_chap_3}) one can verify that $\mathbf{z_{2p}}, \mathbf{z_{2q}} \in \mathcal{L}_{\infty} \Rightarrow \mathbf{p}, \dot{\mathbf{p}}, \mathbf{q}, \dot{\mathbf{q}} \in \mathcal{L}_{\infty}$ as desired trajectories are bounded via Assumption 3. Now, boundedness of $\mathbf{p}, \dot{\mathbf{p}}, \mathbf{q}, \dot{\mathbf{q}}$ imply $\alpha_{\mathbf{p}}, \alpha_{\mathbf{q}} \in \mathcal{L}_{\infty} \Rightarrow\mathbf{\nu_{p}}, \mathbf{\nu_{q}}\in \mathcal{L}_{\infty}$ from (\ref{ctr_p_chap_3}), (\ref{ctr_q_chap_3}) and $\mathbf{C}(\mathbf{q},\dot{\mathbf{q}}) \in \mathcal{L}_{\infty}$ by property of Euler-Lagrange mechanics \cite{spong2008robot}; these cumulative boundedness conditions imply $\mathbf{\eta_{p}},\mathbf{\eta_{q}} \in \mathcal{L}_{\infty}$ from (\ref{un_p_chap_3}) and (\ref{uncer_att_chap_3}). Therefore, there exists finite constant $c$ such that $\|\mathbf{\eta_m}\|\epsilon_m \leq c$. Then, observing the four stability cases it can be inferred that
\begin{align}
    \dot{V} \leq - \varsigma V + c
\end{align}
implying that the closed-loop system remains ultimately bounded (cf. \cite{khalil2002nonlinear} for the definition) via the barrier Lyapunov function $V$, implying $z_{pi}, z_{qi}$ never violates the constraints and  $|z_{pi}| < k_{pi}, |z_{qi}| < k_{qi}$ $i=1,2,3$ $\forall t>0$. 
\end{proof}
\begin{figure}[htb]
    \includegraphics[width=1\textwidth]{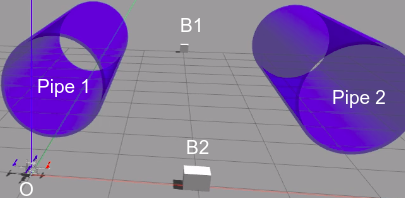}
    \centering
    \caption{Scenario 1 (Pipe) setup}
     \label{fig:pipe_setup}
\end{figure}
\section{Simulation Verification}
In this section, the effectiveness of the proposed RSB (Robust Symmetrical Barrier-lyapunov) controller is compared to the Sliding Mode Control (SMC) \cite{xu2008sliding} during maneuvering in constrained space. The following subsections discuss the simulation setup and results.

\FloatBarrier
\begin{figure}[ht!]
    \includegraphics[width=1\textwidth, height=3in]{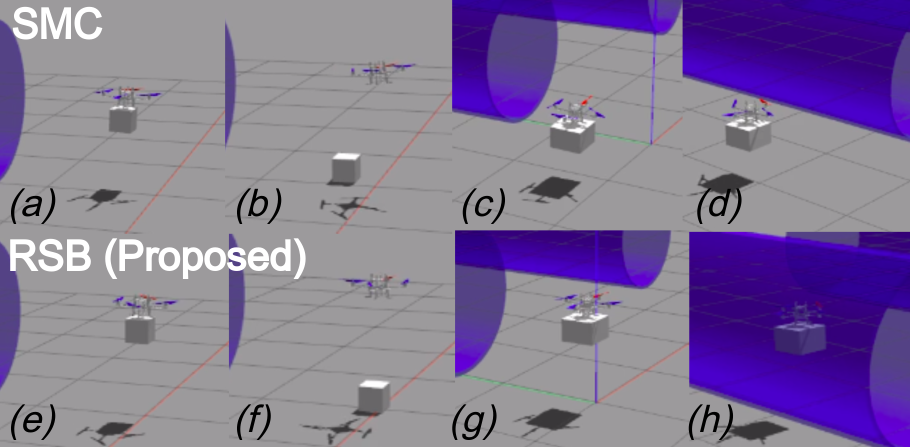}
    \centering
    \caption{Snapshots of the simulation for scenario 1: (a) Before dropping B1 on drop-point (t=34s) (b) After dropping B1 (t=35s) (c) Attempting to enter pipe 1 (t=47s) (d) Failing to enter pipe 1 (t=50s) (e) Before dropping B1 on drop-point (t=34s) (f) After dropping B1 (t=35s) (g) Attempting to enter pipe 1 with B2 (t=47s) (h) After entering pipe 1 (t=51s)}
     \label{fig:pipe_snaps}
\end{figure}

\begin{figure}[ht!]
    \includegraphics[width=1\textwidth]{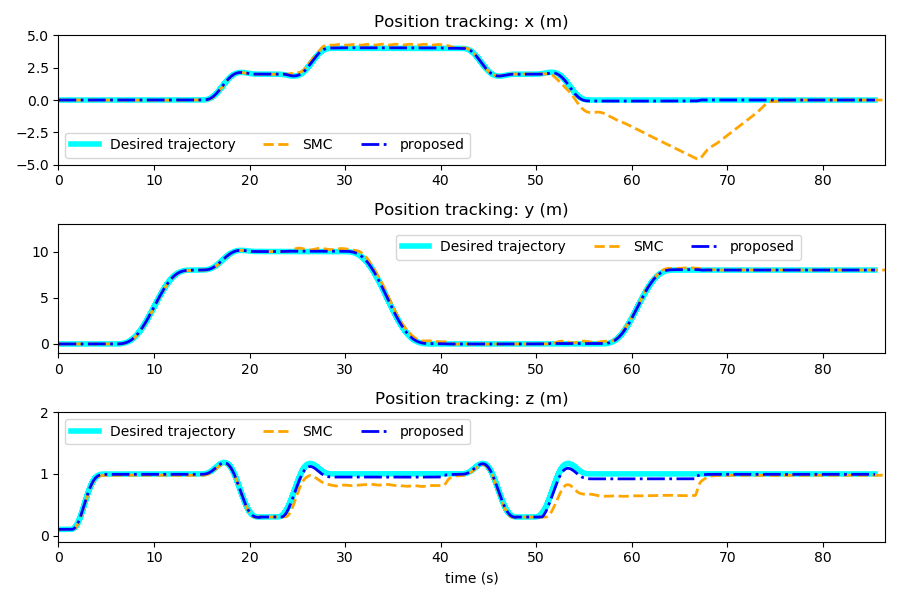}
    \centering
    \caption{Scenario 1 (Pipe): Trajectory tracking}
     \label{fig:pipe_traj}
\end{figure}
\begin{figure}[ht!]
    \includegraphics[width=1\textwidth]{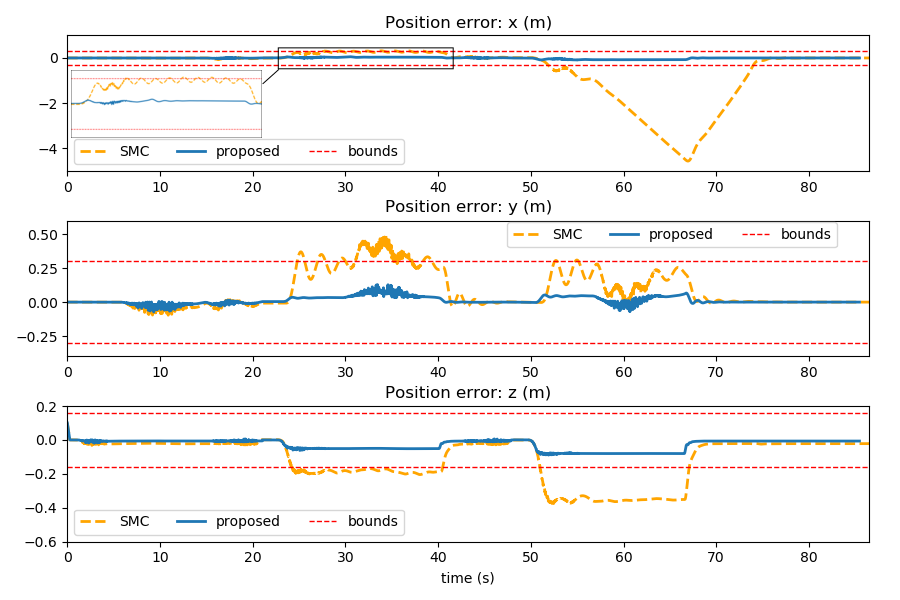}
    \centering
    \caption{Scenario 1 (Pipe): Position error comparison}
    \label{fig:pipe_pos_error}
\end{figure}
\begin{figure}[ht!]
    \includegraphics[width=1\textwidth]{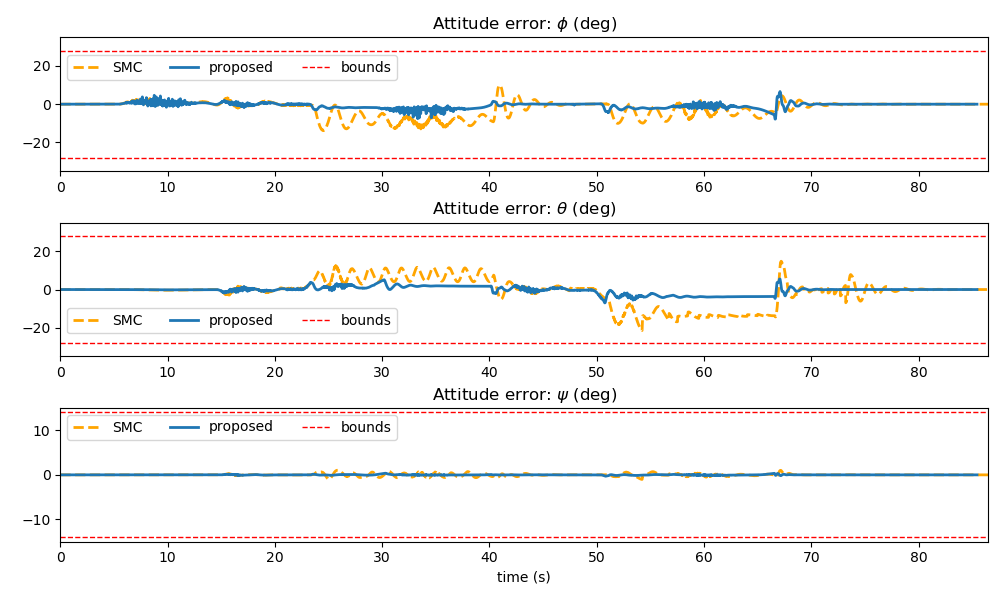}
    \centering
    \caption{Scenario 1 (Pipe): Attitude error comparison}
     \label{fig:pipe_att_error}
\end{figure}
\begin{figure}[ht!]
    \includegraphics[width=1\textwidth]{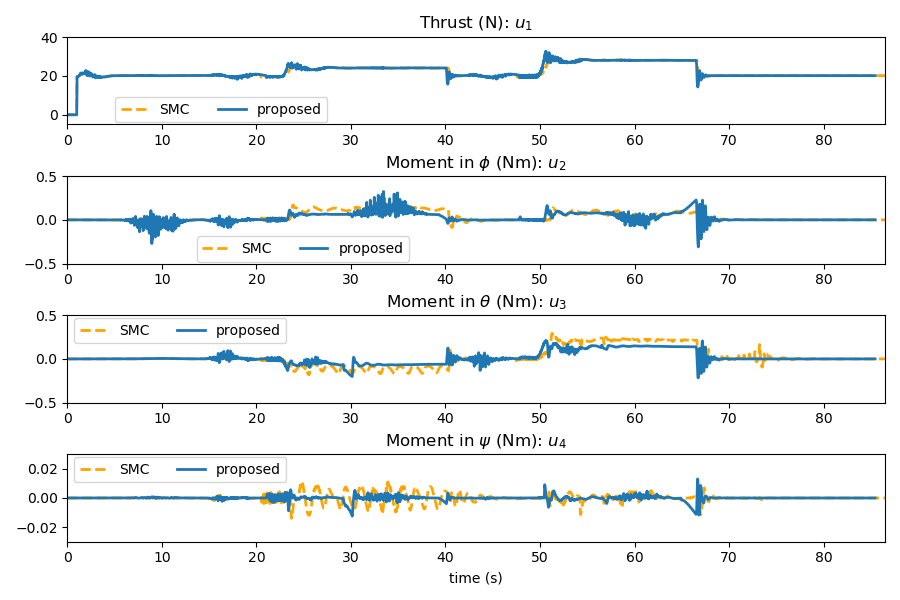}
    \centering
    \caption{Scenario 1 (Pipe): Control inputs}
     \label{fig:pipe_control_input}
\end{figure}


The performance of the proposed controller is tested on a Gazebo simulation platform using the RotorS Simulator framework \cite{rotors_sim_2016} for ROS with the Pelican quadrotor model, which is actuated by commanding the angular velocities, in accordance with their relation with the applied thrusts and moments (cf. \cite{mellinger2011minimum}). The objective is to determine the effectiveness of the proposed controller on precise path tracking and payload delivery of the quadrotor. The following parameters for the model and the proposed controller are used: $[k_{p_1}, k_{p_2}, k_{p_3}, k_{q_1}, k_{q_2}, k_{q_3}]^{T} = [0.3, 0.3, 0.16, 0.5, 0.5, 0.25]^{T}$;  $\Lambda_{1p} = [0.1,~ 0.1,~ 0.1]$; $\Lambda_{2p} = \diag{\{6, 6, 6\}}$; $\Lambda_{1q}= [24,~24,~8]$; $\Lambda_{2q} = \diag{\{15, 15, 15\}}$; $\bar{m}=2$ kg; $\bar{\mathbf{J}}=\diag{\{0.02, 0.02, 0.04\}}$; $E_q = 0.3$;  $E_p = 0.3$; $\epsilon_p = 0.1$; $\epsilon_q = 1$; mean wind-speed disturbance = 2 m/s; $\mathbf{p(0)=q(0)=0}$; disturbance bounds for the outer-loop controller: $[0,0,0.5]^{T}$; disturbance bounds for the inner-loop controller: $[0.01,0.01,0.02]^{T}$. In real-life dynamic payload delivery scenario, the centre-of-mass of the payload might not exactly align with that of the quadrotor, and so, small offsets (valued mentioned in the scenarios) are provided to each payload in the x and y directions for more generalization. For SMC, sliding surfaces are taken as $\mathbf{s_p}=\dot{\mathbf{z}}_{\mathbf{p}} + \Lambda_{1p} \mathbf{z_p}$, $\mathbf{s_q}=\dot{\mathbf{z}}_{\mathbf{q}} + \Lambda_{1q} \mathbf{z_q}$, respectively, for parity.

\textit{Scenario 1 (Pipe):} In this simulation, there are 2 payloads, B1 (0.4 kg) and B2 (0.8 kg); and 2 pipes of 1.5 m diameter each, as shown in Fig. \ref{fig:pipe_setup}. The quadrotor starts from the initial position (O), travels through pipe 1, picks up B1, travels through the pipe 2, drops B1 at its drop point, picks up B2, travels through pipe 1, and finally drops B2 at its drop point. Figures \ref{fig:pipe_pos_error} and \ref{fig:pipe_att_error} show the results on the position and attitude errors are obtained after simulating on SMC and the proposed controller. It is readily observable that, at $t \approx 24$s and at $t \approx 51$s, all three position errors for SMC violate the intended performance-bound imposed (shown by the red horizontal lines), while the proposed controller keeps the position errors within the bounds. While the SMC does not violate the intended attitude error limits, it should be noted that the magnitudes of error for SMC are significantly larger with relatively larger oscillations, which is detrimental to precise payload delivery operations, as seen later. Table \ref{table 1} compares the Root-mean-squared (RMS) and peak errors of the proposed controller and SMC, which emphasises on the points stated earlier. Additionally, to further demonstrate the effectiveness of the proposed controller, Fig. \ref{fig:pipe_snaps} shows 4 essential snapshots, of the simulation with timestamps. 

From Fig. \ref{fig:pipe_snaps}((a),(b) and (e),(f)), it is clearly seen that the quadrotor successfully drops B1 at approximately its designated drop-point (which falls on the red line) for the proposed controller, while it fails to do the same with SMC. Additionally, from Fig. \ref{fig:pipe_snaps}((c),(d) and (g),(h)), it is observed that the quadrotor successfully manoeuvres through the pipe for the proposed controller, unlike with the SMC.


\begin{figure}[htb]
    \includegraphics[width=1\textwidth, height=3.5in]{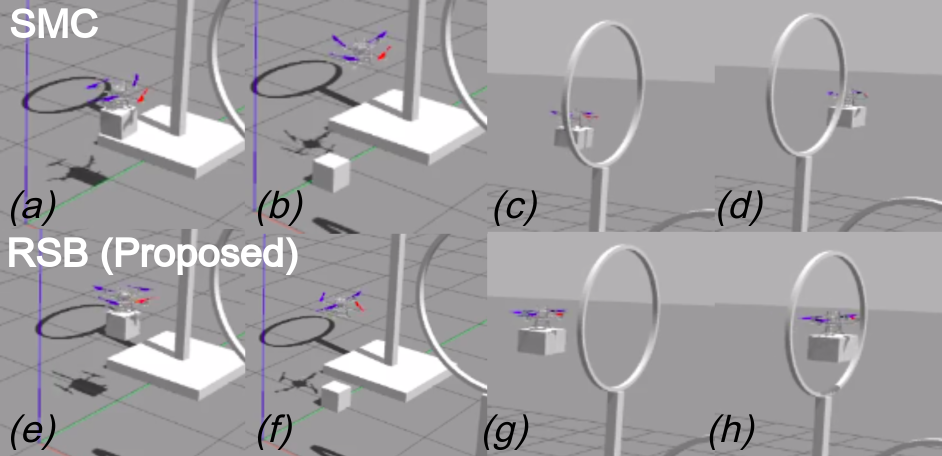}
    \centering
    \caption{Scenario 2 (Ring) simulation run snapshots: (a) Before dropping B1 on drop-point (t=59s) (b) After dropping B1 (t=60s) (c) Attempting to enter ring 2 (t=84s) (d) Failing to enter ring 2 (t=85s) (e) Before dropping B1 on drop-point (t=59s) (f) After dropping B1 (t=60s) (g) Attempting to enter ring 2 with B2 (t=84s) (h) After entering ring 2 (t=85s) }
     \label{fig:ring_snaps}
\end{figure}

\textit{Scenario 2 (Ring): } In this simulation Scenario, (cf. Fig. \ref{fig:ring_setup}), the quadrotor is tasked to pick up payload B1 ($0.4$ kg) and manoeuvre through four ring-like structures of different diameters ($1.2$ m for rings 1-2 and $1.5$ m for rings 3-4) for $t=60$ sec before dropping at its original place. Then, the similar task is to be repeated from $t >60$ sec with a higher payload B2 ($0.8$ kg). 
For realistic representation, (i) wind disturbance of speed $2$ m/sec blowing in $45^{0}$ direction to $x-y$ plane is added; (ii) a $2$ cm offset in centre-of-mass between the payload and Pelican in $x$ direction is provided, as such misalignment cannot be ruled out in practice. The tracking performance of the two controllers are shown via Figs. \ref{fig:ring_snaps}-\ref{fig: ctr_ip_ring} and Table \ref{table 1}. 
It can observed from Fig. \ref{fig:ring_pos_error} that, while carrying payload B1, $z$ position error violates the intended performance-bound with SMC and such violations are quite significant in $y$ and $z$ directions while carrying the heavier payload B2 (cf. $88 \leq t \leq 115$ sec (approx)): these violations imply that the quadrotor did not manoeuvre through the rings with SMC while carrying payload B2 (cf. Fig. \ref{fig:ring_snaps}, peak error in Table \ref{table 2}). Whereas, the proposed controller keeps the position and attitude errors within the designated bounds. While the SMC does not violate the intended attitude error limits, Table \ref{table 2} reveals performance gaps in terms of root-mean-squared (RMS) and peak absolute errors.  While the SMC does not violate the intended attitude error limits, it should be noted that the magnitudes of error are significantly larger with SMC, with relatively larger oscillations, which is detrimental to precise and fast payload delivery operations. In fact, to further demonstrate the effectiveness of the proposed controller for payload operation scenarios, Fig. \ref{fig:ring_snaps} shows four essential snapshots for each controller method. It is clearly seen in Fig. \ref{fig:ring_snaps} (a), (b), (e) and (f) that the quadrotor successfully drops B1 at approximately its designated drop-point (which is situated at the intersection of the green line and the relevant grey grid-line) for the proposed controller, while it fails to do the same with SMC.

\begin{figure}[ht!]
{\includegraphics[width=6in, height=3.5in]{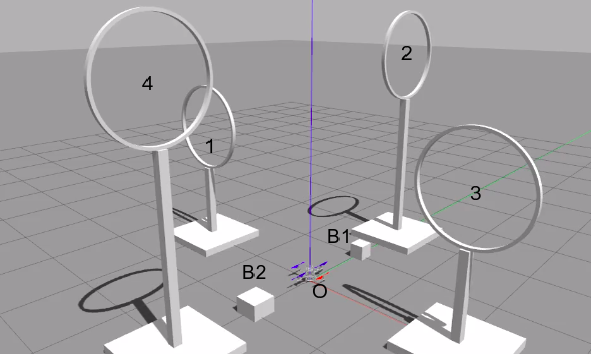}}
        \centering
        \caption{Simulation scenario 2: manoeuvring through ring structures with different payloads.}
        \label{fig:ring_setup}
    \end{figure}
\begin{figure}[ht!]
    \includegraphics[width=6in, height=4in]{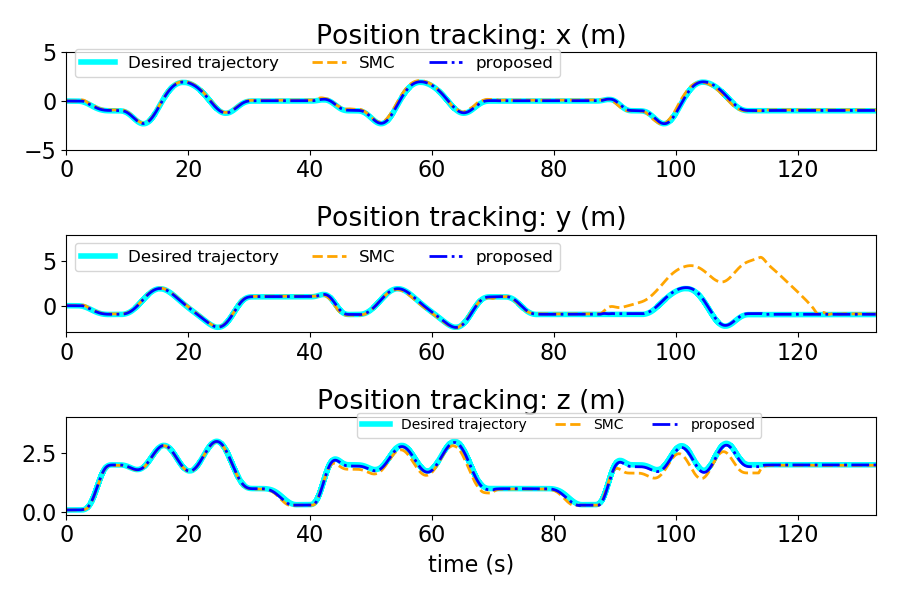}
    \centering
    \caption{Scenario 2 (Ring): Position tracking performance comparison}
    \label{fig:ring_pos}
\end{figure}
\begin{figure}[ht!]
    \includegraphics[width=6in, height=3.8in]{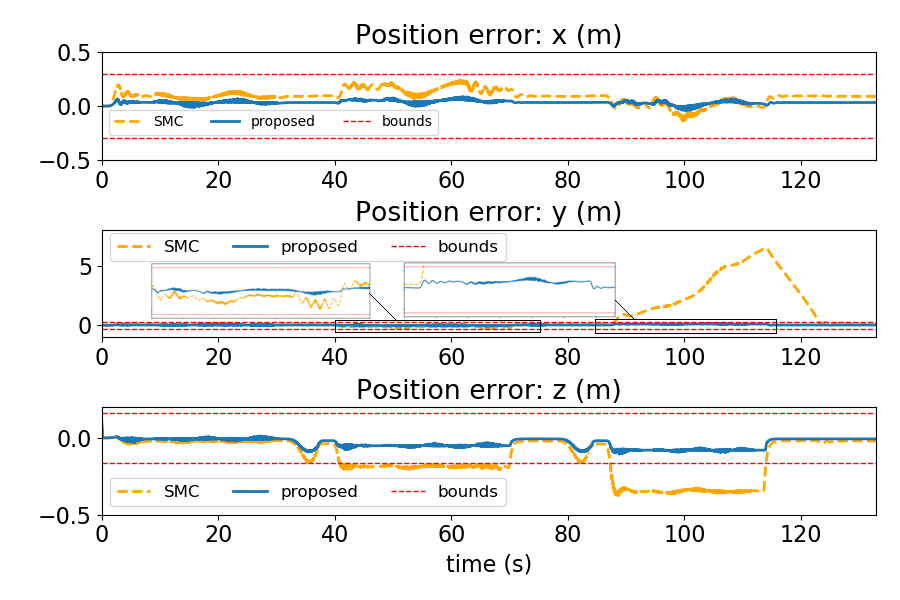}
    \centering
    \caption{Scenario 2 (Ring): Position tracking error comparison for the ring scenario}
    \label{fig:ring_pos_error}
\end{figure}
\clearpage
\FloatBarrier
\begin{figure}[ht!]
    \includegraphics[width=6in, height=3.8in]{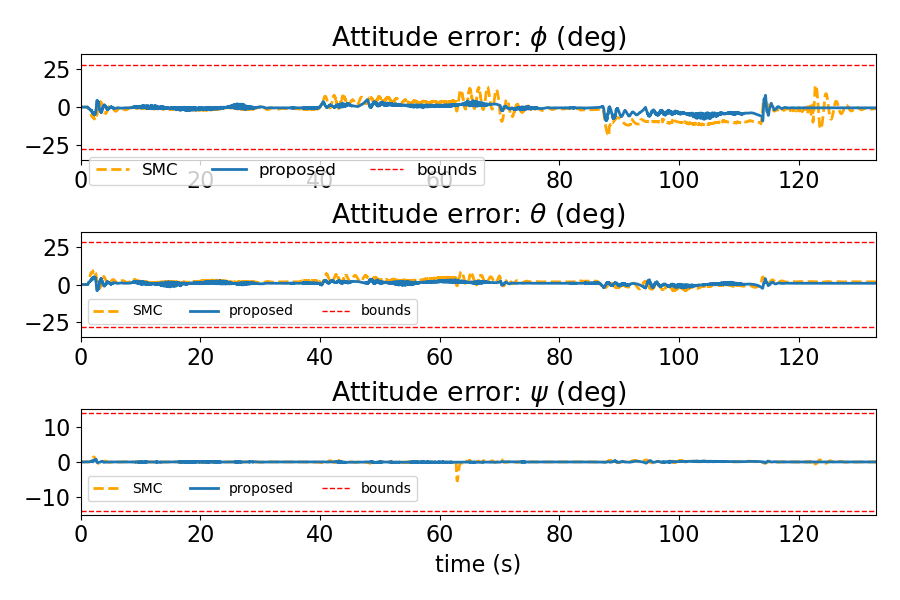}
    \centering
    \caption{Scenario 2 (Ring): Attitude tracking error comparison}
    \label{fig:ring_att_error}
\end{figure}
\begin{figure}[ht!]
    \includegraphics[width=6in, height=3.8in]{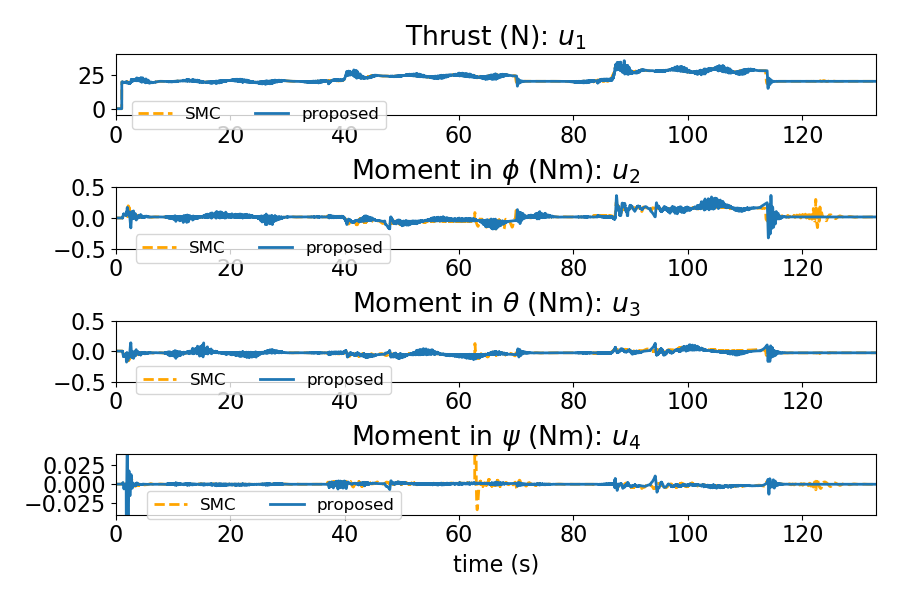}
    \centering
    \caption{Scenario 2 (Ring): Control input comparison}
    \label{fig: ctr_ip_ring}
\end{figure}

\begin{table}[!t]
\renewcommand{\arraystretch}{1.6}
		\centering
{
{	\begin{tabular}{|c| c c c c c c|}
		\hline
		& \multicolumn{3}{c}{RMS error (m)} & \multicolumn{3}{c|}{RMS error (degree)} \\ \cline{1-7}
		 & $x$ & $y$  & $z$  & $\phi$ & $\theta$  & $\psi$  \\
		 \hline
		SMC & 1.28 & 0.16  & 0.17 & 4.71 & 7.09  & 0.25 \\
		\hline
		Proposed & {0.03} & {0.03}  & {0.04} & 1.55 & 1.94  & 0.06 \\
		\hline
		& \multicolumn{3}{c}{Peak absolute error (m)} & \multicolumn{3}{c|}{Peak absolute error (degree)}  \\ \cline{1-7}
		 & $x$ & $y$  & $z$  & $\phi$ & $\theta$  & $\psi$  \\
		 \hline
		SMC & 4.61 & 0.49  & 0.38 & 10.61 & 21.25  & 0.81 \\
		\hline
		Proposed & {0.0} & {0.12}  & {0.10} & 7.11 & 6.08  & 0.46 \\
		\hline
\end{tabular}}}
\caption{{Scenario 1 (Pipe): Performance comparison}} \label{table 1}
\end{table}
\clearpage
\begin{table}[t]
\renewcommand{\arraystretch}{1.6}
		\centering
{
{	\begin{tabular}{|c| c c c c c c|}
		\hline
		Controller & \multicolumn{3}{c}{RMS error (m)} & \multicolumn{3}{c|}{RMS error (degree)}  \\ \cline{1-7}
		 & $x$ & $y$  & $z$  & $\phi$ & $\theta$  & $\psi$  \\
		 \hline
		SMC & 0.11 & 1.89  & 0.18 & 5.58 & 2.55  & 0.28  \\
		\hline
		Proposed & {0.04} & {0.05}  & {0.04} & 2.16 & 1.13  & 0.07 \\
		\hline
		& \multicolumn{3}{c}{Peak absolute error (m)} & \multicolumn{3}{c|}{Peak absolute error (degree)}  \\ \cline{1-7}
		 & $x$ & $y$  & $z$  & $\phi$ & $\theta$  & $\psi$  \\
		 \hline
		SMC & 0.24 & 6.47  & 0.38 & 18.69 & 9.53  & 0.72 \\
		\hline
		Proposed & {0.09} & {0.16}  & {0.11} & 9.24 & 4.99  & 0.72 \\
		\hline
\end{tabular}}}
\caption{{Scenario 2 (Ring): Performance comparison}}\label{table 2}
\end{table}

\section{Conclusion}
A robust controller with Barrier Lyapunov approach for quadrotors was proposed, which can ensure tracking performance within a predefined limit to allow a quadrotor navigate through constrained space under parametric uncertainties and external disturbances. Closed-loop system stability was established via Barrier Lyapunov function. Performance of the proposed controller is extensively verified via simulated scenarios in comparison with the state of the art. However, the six states, i.e., the linear velocities and the angular velocities are unconstrained in this work. The next chapter of the thesis implements robust control with full state-constraints (i.e., constraints on position, orientation, linear velocity and angular velocity) on a quadrotor and a payload delivery scenario under tight constraints is extensively investigated via experiments on a real quadrotor.


\chapter{Robust Manoeuvring of Quadrotor under Full State Constraints}
\label{ch:chap4}
\section{Introduction}

In the previous chapter, the problem of robust manoeuvring and payload carrying by a quadrotor under spatial constraints was addressed. However, six out of the twelve controllable states of the quadrotor, i.e., the three linear velocity and the three angular velocity states were left unconstrained. In real-life scenarios, such as relief operation during flooding, the quadrotor needs to not only manoeuvre under tight spatial constraints, but also needs to drop the relief package with high precision to avoid wastage of material. On the other hand (further observed in the experimental section of this chapter), a quadrotor has the tendency to undergo heavy transients immediately after deploying its payload while tracking the desired trajectory. In conventional robust or adaptive control schemes, these transients are controlled via tuning of the user-specified control gains. A drawback of this method is that tuning of a real quadrotor requires several trial-and-error runs in real-life scenarios (while keeping in mind the effect of parametric uncertainties and external disturbances) before ensuring any predefined accuracy and successful performance under implied state constraints. This defeats the purpose of deploying the UAV in emergency scenarios which requires quick and effective response. Therefore, a control design enforcing user-defined constraints on the time-derivatives of position and orientation implies that the quadrotor would perform better at mitigating transients during payload drop operations. The additional challenge of the quadrotor being an underactuated system adds to the challenge of accurate trajectory tracking.

The challenges faced by related works are explored in Chapter 1, where a comprehensive control solution to tackle the above issues, to the best of the author's knowledge, was found missing. 
Toward this direction, this chapter provides the following major contributions:

\begin{itemize}
    \item A BLF-based robust controller for quadrotor is formulated which obeys full state-constraints on all six degrees-of-freedom (DoFs) (i.e., position, orientation, linear and angular velocity), and can tackle parametric uncertainties and external disturbances. Compared to \cite{dasgupta2019singularity}, constraints on all DoFs make the control problem more practical. 
    \item The closed-loop stability is verified analytically, and the performance of the proposed design is experimentally compared to the state-of-the-art designs.
    
The rest of the chapter is organised as follows: Section 4.2 describes the quadrotor dynamics, Section 4.3 states the control problem; Section 4.4 details the proposed control framework, and the stability analysis is provided in Section 4.5; experimental results are provided in Section 4.6, while Section 4.7 provides the concluding remarks.

\end{itemize}

\section{Quadrotor Dynamics}
Similar to Chapter 3, the dynamics of a quadrotor is modelled on the Euler-Lagrange dynamics as follows \cite{Bialy2013RobustAdaptive_UAV}:
\begin{align}
    m\Ddot{\mathbf{p}}+\mathbf{g+d_{p}} &= \mathbf{\tau_{p}} \label{eq:tau_p_chap_4} \\
    \mathbf{J(q)}\Ddot{\mathbf{q}}+\mathbf{C(q,}\Dot{\mathbf{q}})\Dot{\mathbf{q}}+\mathbf{d_{q}} &= \mathbf{\tau_{q}} \label{eq:tau_q_chap_4}\\
    \mathbf{\tau_p} &= \mathbf{R}_B^{W}\mathbf U \label{comb_chap_4}
\end{align}

\noindent where $m$ is the total mass of the system; $\mathbf{p}(t) \triangleq  \begin{bmatrix}
    x(t) & y(t) & z(t)
\end{bmatrix}^T \in \mathbb{R}^{3}$ is the position of the centre of mass of the quadrotor in the Earth-fixed frame; $\mathbf{q}(t) \triangleq  \begin{bmatrix}
    \phi(t) & \theta(t) & \psi(t)
\end{bmatrix}^T \in \mathbb{R}^{3}$ is the attitude vector consisting of the roll ($\phi$), pitch ($\theta$) and yaw ($\psi$) angles; $\mathbf{g} \triangleq  \begin{bmatrix}
    0 & 0 & mg
\end{bmatrix}^T \in \mathbb{R}^{3}$, where $g$ is the acceleration due to gravity in the z-direction; $\mathbf{J(q)} \in \mathbb{R}^{3 \times 3}$ is the inertia matrix; $\mathbf{C(q,\dot{q})} \in \mathbb{R}^{3 \times 3}$ is the Coriolis matrix and the vectors $\mathbf{d_p}, \mathbf{d_q} \in \mathbb{R}^{3}$ represent effect of external disturbances (e.g., wind, gust); $\boldsymbol \tau_{\mathbf q} \triangleq 
	\begin{bmatrix}
	u_2(t) & u_3(t) & u_4(t)
	\end{bmatrix}^T\in\mathbb{R}^3$ denotes the control inputs for roll, pitch and yaw; $\boldsymbol \tau_{\mathbf p}(t) \in \mathbb{R}^3$ is the generalized control input for position tracking in Earth-fixed frame, with $\mathbf U(t)\triangleq
	\begin{bmatrix}
	0 & 0 & u_1(t)
	\end{bmatrix}^T\in \mathbb{R}^3$ being the force vector in body-fixed frame and $\mathbf R_B^W \in\mathbb{R}^{3\times3}$ being the $Z-Y-X$ Euler angle rotation matrix describing the rotation from the body-fixed coordinate frame to the Earth-fixed frame, given by
	\begin{align}
	\mathbf R_B^W =
	\begin{bmatrix}
	c_{\psi}c_{\theta} & c_{\psi}s_{\theta}s_{\phi} - s_{\psi}c_{\phi} & c_{\psi}s_{\theta}c_{\phi} + s_{\psi}s_{\phi} \\
	s_{\psi}c_{\theta} & s_{\psi}s_{\theta}s_{\phi} + c_{\psi}c_{\phi} & s_{\psi}s_{\theta}c_{\phi} - c_{\psi}s_{\phi} \\
	-s_{\theta} & s_{\phi}c_{\theta} & c_{\theta}c_{\phi}
	\end{bmatrix}, \label{rot_matrix_chap_4}
	\end{align}
where $c_{(\cdot)} , s_{(\cdot)}$ and denote $\cos{(\cdot)} , \sin{(\cdot)}$ respectively.

\section{Control Problem}

Manoeuvring under full-state constraints implies the quadrotor cannot go beyond some performance limits over the desired trajectories, transforming the control objective as tracking with predefined accuracy.

Let $\mathbf{z_{p}}=\lbrace z_{p1},~z_{p2},~z_{p3} \rbrace$ and $\mathbf{z_{q}}=\lbrace z_{q1},~z_{q2},~z_{q3} \rbrace$ be the position $(x,y,z)$ and attitude $(\phi
, \theta, \psi)$ tracking errors respectively; $\dot{\mathbf{z}}_\mathbf{p}=\lbrace \dot{z}_{p1},~\dot{z}_{p2},~\dot{z}_{p3} \rbrace$ and $\dot{\mathbf{z}}_\mathbf{q}=\lbrace \dot{z}_{q1},~\dot{z}_{q2},~\dot{z}_{q3} \rbrace$ be the tracking errors in velocity $(\dot{x},\dot{y},\dot{z})$ and attitude rate $(\dot{\phi}
, \dot{\theta}, \dot{\psi})$ respectively. Let, $k_{pi}, k_{qi}, \dot{k}_{pi}, \dot{k}_{qi} \in \mathbb{R}^{+}$ for $i=1,~2,~3$ be twelve user-defined scalars. Then, the control problem is defined as follows:

Under Assumptions 1-2 of Section 3.2.1, to design a robust controller to track a desired trajectory such that the tracking accuracy for position and velocity remains within a user-defined region as $|z_{pi}|< k_{pi}, |z_{qi}|< k_{qi}$, $|\dot{z}_{pi}|< \dot{k}_{pi}, |\dot{z}_{qi}|< \dot{k}_{qi}$, with $i=1,2,3$ and $|z_{pi}(0)|< k_{pi},~|\dot{z}_{pi}(0)|< \dot{k}_{pi},~ |z_{qi}(0)|< k_{qi},~   |\dot{z}_{qi}(0)|< \dot{k}_{qi}$, $\forall t > 0$. 


Initial conditions lying within the constraints as above is standard in literature, otherwise the control problem becomes ill-posed and infeasible (cf. \cite{KPTee2009BLF, liu2016barrier,dasgupta2019singularity,sachan2019output}). The following section provides a solution to this control problem.

\section{Proposed Control Solution}
 The we follow the same co-design approach of simultaneous controller design of an outer loop for position dynamics and of an inner loop for attitude dynamics as in the previous chapter.

\subsection{Outer Loop Controller Design}
The position tracking error  $\mathbf{z_{p}}=\lbrace z_{p1},~z_{p2},~z_{p3} \rbrace$ and an auxiliary error variable $\mathbf{z_{2p}}$ are defined as
\begin{align}
    \mathbf{z_{p}} &= \mathbf{p-p_{d}}, \label{eq:z_p} \\
    \mathbf{z_{2p}} &= \Dot{\mathbf{p}}-\mathbf{\alpha_{p}}, \label{eq:z_2p}
\end{align}
where $\mathbf{\alpha_{p}}$ is an auxiliary control variable. The position control law is designed as
\begin{subequations}\label{ctr_p}
\begin{align}
    \mathbf{\tau_{p}}&=\Bar{m}\mathbf{\nu_{p}}+\bar{\mathbf{g}},~~\mathbf{\nu_{p}} = \bar{\mathbf{\nu}}_{\mathbf p}+\Delta\mathbf{\nu_{p}}, \label{eq:tau_p_expression} \\
    \mathbf{\alpha_{p}} &=-\Lambda_{1p}\mathbf{z_{p}}+\dot{\mathbf p}_{\mathbf d}, \label{alpha_p}\\
    \bar{\mathbf{\nu}}_{\mathbf p} &= \dot{\mathbf \alpha}_{\mathbf p}-\mathbf{D_{3p}}\mathbf{z_{p}}-\Lambda_{2p}\mathbf{z_{2p}},\label{eq:nu_p expression} \\
    \Delta \mathbf{\nu_{p}}&=-\rho_{p}\sat(\mathbf{z_{2p}}, \epsilon_p), \label{eq:delta_nu_p}
\end{align}
\end{subequations}
where $\bar{\mathbf{g}}=[0~0~\bar{m}g]^T$; $\Lambda_{1p} = \diag{\{\gamma_{p_{1}}, \gamma_{p_{2}},\gamma_{p_{3}}\}}$, and $\Lambda_{2p} = \diag{\{\gamma_{2p_{1}}, \gamma_{2p_{2}},\gamma_{2p_{3}}\}}$, where $\gamma_{p_{i}}>0$ and  $\gamma_{2p_{i}}>0$ are user-defined constants; $\mathbf{D_{3p}} \triangleq \mathbf D_{\mathbf{2p}}^{-1}\mathbf{D_{p}}$ with \\ $\mathbf{D_{p}}=\diag{\bigg\{\frac{1}{k_{p_{1}}^{2}- z_{p_{1}}^{2}}, \frac{1}{k_{p_{2}}^{2}- z_{p_{2}}^{2}}, \frac{1}{k_{p_{3}}^{2}- z_{p_{3}}^{2}}\bigg\}}$, \\ $\mathbf{D_{2p}}=\diag{\bigg\{\frac{1}{k_{2p_{1}}^{2}- z_{2p_{1}}^{2}}, \frac{1}{k_{2p_{2}}^{2}- z_{2p_{2}}^{2}}, \frac{1}{k_{2p_{3}}^{2}- z_{2p_{3}}^{2}}\bigg\}}$ and 
\begin{align}
    k_{2p_i} &\triangleq \dot{k}_{p_i} + \gamma_{p_i}k_{p_i},~i=1,2,3. \label{eq: k_2p def}
\end{align}
The gain $\rho_{p}$ (defined later) provides robustness against uncertainties and $\epsilon_p>0$ is used to avoid chattering.

\noindent From \eqref{eq:z_2p} and \eqref{alpha_p}, $z_{2p_i}$ can be rewritten as
\begin{align*}
    z_{2p_i} = \dot{z}_{p_i}+\gamma_{p_i} z_{p_i} .
\end{align*}
Then, based on the desired constraints $|z_{pi}|< k_{pi}, |\dot{z}_{pi}|< \dot{k}_{pi}$ and using (\ref{eq: k_2p def}), the desired constraint on $z_{2p_i}$ turns out to be 
\begin{align}
    |z_{2p_i}| &\leq |\dot{z}_{p_i}|+|\gamma_{p_i}||z_{p_i}|  < \dot{k}_{p_i} + \gamma_{p_i}k_{p_i} = k_{2p_i}. \label{new}
\end{align}
%
%
%
Using (\ref{eq:tau_p_chap_4}) and \eqref{eq:tau_p_expression}, the time derivative of \eqref{eq:z_2p} yields
\begin{align}
    \Dot{\mathbf z}_{2\mathbf p} &= \Ddot{\mathbf{p}}-\Dot{\mathbf \alpha}_{\mathbf p} = m^{-1}(\mathbf{\tau_{p}}-\mathbf{g}-\mathbf{d_{p}})-\Dot{\mathbf \alpha}_{\mathbf p} \nonumber \\ 
    &= \Bar{\nu}_{\mathbf p}+\Delta \mathbf{\nu_{p}}+\mathbf{\eta_{p}}-\Dot{\mathbf \alpha}_{\mathbf p} \label{ddot_z2p}\\
\text{where} ~~~ \mathbf{\eta_{p}} &\triangleq (m^{-1}\bar{m}-1)\mathbf{\nu_{p}}+m^{-1}(\bar{\mathbf{g}}-\mathbf{g}-\mathbf{d_p}) \label{un_p}
\end{align}  
is the \textit{overall uncertainty} in the position dynamics and $\rho_p$ is designed using (\ref{mass}) and (\ref{un_p}) such that $\rho_p \geq ||\mathbf{\eta_{p}}|| $ as 
\begin{align}
    \rho_{p} &\geq E_{p} \|\Bar{\mathbf{\nu}}_{\mathbf p}\|+ E_{p} \rho_{p} +m^{-1} ||(\bar{\mathbf{g}}-\mathbf{g}-\mathbf{d_p}) || \nonumber\\
   \Rightarrow \rho_{p} &\geq \ddfrac{1}{1-E_{p}}(E_{p}\|\Bar{\mathbf{\nu}}_{\mathbf p}\|+m^{-1}(||\bar{\mathbf{g}}-\mathbf{g}||+ ||\mathbf{d_p}||). \label{eq:rho_p}
\end{align}
Eventually, $\mathbf{U}$ is applied to the system via (\ref{comb_chap_4}) as $R_B^W$ is an invertible rotation matrix.  
\subsection{Inner Loop Controller Design}
In a similar approach to the previous chapter, the desired roll ($\phi_d$) and pitch ($\theta_d$) angles are defined by first defining an intermediate coordinate frame as (cf. \cite{mellinger2011minimum}):
\begin{subequations}\label{int_co}
\begin{align}
    z_B &= \frac{\mathbf{\tau_{p}}}{||\mathbf{\tau_{p}}||} ,~~y_A = \begin{bmatrix}
    -s_{\psi_d} & c_{\psi_d} & 0
\end{bmatrix}^T \\
    x_B &= \frac{y_A \times z_B}{||y_A \times z_B||} ,~~y_B = z_B \times x_B
\end{align}
\end{subequations}
where $(x_B,y_B,z_B)$ are the desired $(x,y,z)$ axis in the body-fixed coordinate frame. Given the desired yaw angle $\psi_d (t)$ and based on the computed intermediate axes as in (\ref{int_co}), $\phi_d (t)$ and $\theta_d (t)$ can be determined using the desired body frame axes as described in \cite{mellinger2011minimum}.

Further, the attitude error can be defined as \cite{mellinger2011minimum}:
\begin{align}
     \mathbf{z_{q}}= \lbrace z_{q1}, z_{q2}, z_{q3} \rbrace &= {\mathbf{((R_d)^T R_B^W - (R_B^W)^T R_d)}}^{v} \label{eq: z_q}
\end{align}
where $(.)^v$ is \textit{vee} map converting elements of $SO(3)$ to $\in{\mathbb{R}^3}$ (cf. \cite{mellinger2011minimum}) and $\mathbf{R_d}$ is the rotation matrix as in (\ref{rot_matrix_chap_4}) evaluated at ($\phi_d, \theta_d, \psi_d$). The  auxiliary error variable $\mathbf{z_{2q}}$ is defined as
\begin{align}
    \mathbf{z_{2q}} &= \Dot{\mathbf{q}}-\mathbf{\alpha_{q}} \label{eq:z_2q}
\end{align}
with $\mathbf{\alpha_{q}}$ being an auxiliary control variable defined subsequently.

The inner loop control law is designed as
\begin{subequations}\label{ctr_q}
\begin{align}
    \mathbf{\tau_{q}}&=\Bar{\mathbf{J}}\mathbf{\nu_{q}}+\bar{\mathbf{C}}\dot{\mathbf{q}},~\mathbf{\nu_{q}} = \Bar{\mathbf{\nu}}_{\mathbf q}+\Delta\mathbf{\nu_{q}}, \label{eq:tau_q_expression}\\
    \mathbf{\alpha_{q}} &=-\Lambda_{1q}\mathbf{z_{q}}+\dot{\mathbf{q}}_\mathbf{d}, \label{alpha_q}\\
    \bar{\mathbf{\nu}}_{\mathbf q} &= \dot{\mathbf{\alpha}}_{\mathbf q}-\mathbf{D_{3q}}\mathbf{z_{q}}-\Lambda_{2q}\mathbf{z_{2q}},\label{eq: nu_q expression} \\
    \Delta \mathbf{\nu_{q}}&=-\rho_{q}\sat(\mathbf{z_{2q}}, \epsilon_q), \label{eq:delta_nu_q}
\end{align}
\end{subequations}
where $\Lambda_{1q} = \diag{\{\gamma_{q_{1}}, \gamma_{q_2},\gamma_{q_3}\}}$, $\Lambda_{2q} = \diag{\{\gamma_{2q_{1}}, \gamma_{2q_2},\gamma_{2q_3}\}}$, with $\gamma_{q_i}>0$ and $\gamma_{2q_i}>0$ being user-defined constants; $\mathbf{D_{3q}} \triangleq \mathbf{D_{2q}}^{-1}\mathbf{D_{q}}$ with $\mathbf{D_{q}}=\diag{\bigg\{\frac{1}{k_{q_{1}}^{2}-z_{q_{1}}^{2}}, \frac{1}{k_{q_{2}}^{2}-z_{q_{2}}^{2}}, \frac{1}{k_{q_{3}}^{2}-z_{q_{3}}^{2}}\bigg\}}$, $\mathbf{D_{2q}}=\diag{\bigg\{\frac{1}{k_{2q_{1}}^{2}-z_{2q_{1}}^{2}}, \frac{1}{k_{2q_{2}}^{2}-z_{2q_{2}}^{2}}, \frac{1}{k_{2q_{3}}^{2}-z_{2q_{3}}^{2}}\bigg\}}$ and
\begin{align}
    k_{2q_i} \triangleq \dot{k}_{q_i} +\gamma_{q_i}k_{q_i} ~i=1,2,3.
    \label{eq: k_2q def}
\end{align}
The gain $\rho_{q}$ (defined later) robustness against uncertainties in the attitude dynamics and $\epsilon_q>0$ avoids chattering. 

Similar to the desired constraint on $z_{2p_i}$ as in (\ref{new}), a desired constraint $|z_{2q_i}| < k_{2q_i}$ is obtained from (\ref{eq: k_2q def}). Following similar lines to derive (\ref{ddot_z2p}) for the outer loop controller, the following is achieved using (\ref{eq:tau_q_chap_4}), \eqref{eq:z_2q} and \eqref{eq:tau_q_expression}
\begin{align}
    \Dot{\mathbf z}_{2\mathbf q} &= \Bar{\nu}_{\mathbf q}+\Delta \mathbf{\nu_{q}}+\mathbf{\eta_{q}}-\Dot{\mathbf \alpha}_{\mathbf q}\label{eq:ddot_z2q}\\
\text{where} ~~~ \mathbf{\eta_{q}} &\triangleq (\mathbf{J}^{-1}\bar{\mathbf J}-\mathbf{I})\mathbf{\nu_{q}}+\mathbf{J}^{-1}(-\mathbf{d_q}-\Delta \mathbf{C}\dot{\mathbf{q}}) \label{uncer_att}
\end{align}
represents the \textit{overall uncertainty} in the attitude dynamics and the robust control gain is designed to be $\rho_q \geq || \mathbf{\eta_{q}} || $, i.e.,
\begin{align}
\rho_{q} &\geq E_q || \Bar{\mathbf{\nu}}_{\mathbf q} || + E_q \rho_q +  || \mathbf{J}^{-1}(-\mathbf{d_q}-\Delta \mathbf{C}\dot{\mathbf{q}})|| \nonumber \\
    \rho_{q} &\geq \ddfrac{1}{1-E_q}(E_q || \Bar{\mathbf{\nu}}_{\mathbf q} || + || \mathbf{J}^{-1} || ( ||\Delta \mathbf{C}||||\dot{\mathbf{q}}|| + || \mathbf{d_q}||   ).  \label{eq:rho_q}
\end{align}
\begin{remark}[Choice of gains and feasibility]
It may appear from (\ref{eq:nu_p expression}) and (\ref{eq: nu_q expression}) that the gains $\mathbf{D_p,~D_q,~D_{2p}}$ and  $\mathbf{D_{2q}}$ will become infeasible if ${z_{pi}}=k_{pi}$, ${z}_{qi}=k_{qi}$, ${z_{2p_i}}=k_{2p_i}$, and ${z}_{2q_i}=k_{2q_i}$ for any $i=1,2,3$ and $t>0$. However, the subsequent closed-loop stability analysis will show that $|{z_{pi}}|<k_{pi}$ and $|{z}_{qi}|<k_{qi}$, $|{z_{2p_i}}|<k_{2p_i}$ and $|{z}_{2q_i}|<k_{2q_i}$ $\forall t>0$ and infeasibility is avoided under the proposed design. Further, smaller $k_{pi},k_{qi},\dot{k}_{pi},\dot{k}_{qi}$ will result in better accuracy, albeit at the cost of higher control input demand. Therefore, such choices should be made as per application requirements.
\end{remark}

\section{Stability Analysis}
\begin{theorem}
Under Assumptions 1-3 of Section 3.2.1 (in Chapter 3) and initial conditions $|z_{pi}(0)|< k_{pi},~|\dot{z}_{pi}(0)|< \dot{k}_{pi},~ |z_{qi}(0)|< k_{qi},~   |\dot{z}_{qi}(0)|< \dot{k}_{qi}$ and using the robust control laws (\ref{ctr_p}) and (\ref{ctr_q}), the tracking error trajectories $z_{pi},~z_{qi},~\dot{z}_{pi},~\dot{z}_{qi}$ remain bounded as $|z_{pi}|< k_{pi}, |z_{qi}|< k_{qi}$, $|\dot{z}_{pi}|< \dot{k}_{pi}, |\dot{z}_{qi}|< \dot{k}_{qi}$ with $i=1,2,3$ $\forall t > 0$.
\end{theorem}
\begin{proof}
Closed-loop stability analysis is carried via the following barrier Lyapunov function candidate:
\begin{align}
    V&=V_p + V_q \label{eq:lyap} \\
\text{where}~~    V_p &= \ddfrac{1}{2}\sum_{i=1}^{3}\log{\left(\ddfrac{k_{p_{i}}^2}{k_{p_{i}}^{2}-z_{p_i}^{2}}\right)}+\log{\left(\ddfrac{k_{2p_{i}}^2}{k_{2p_{i}}^{2}-z_{2p_i}^{2}}\right)} \nonumber \\
    V_q &= \ddfrac{1}{2}\sum_{i=1}^{3}\log{\left(\ddfrac{k_{q_{i}}^2}{k_{q_{i}}^{2}-z_{q_i}^{2}}\right)}+ \log{\left(\ddfrac{k_{2q_{i}}^2}{k_{2q_{i}}^{2}-z_{2q_i}^{2}}\right)}. \nonumber
\end{align}

\noindent Using (\ref{eq:z_2p})-\eqref{ddot_z2p} and (\ref{eq:z_2q})-\eqref{eq:ddot_z2q}, with gains $\mathbf{D_p}, \mathbf{D_{2p}}$ as in (\ref{eq:nu_p expression}) and $\mathbf{D_q}, \mathbf{D_{2q}}$ as in (\ref{eq: nu_q expression}), the time derivative of \eqref{eq:lyap} yields
\begin{align}
    &\dot{V}_p = \sum_{i=1}^{3}\ddfrac{z_{p_i}^{T}\dot{z}_{p_i}}{k_{p_i}^{2}-z_{p_i}^{2}} +  \ddfrac{z_{2p_i}^{T}\dot{z}_{2p_i}}{k_{2p_i}^{2}-z_{2p_i}^{2}}= \mathbf{z}^{T}_{\mathbf p} \mathbf{D_p}\dot{\mathbf z}_{\mathbf p} + \mathbf{z}^{T}_{\mathbf{2p}}\mathbf{D_{2p}}\dot{\mathbf z}_{\mathbf{2p}} \nonumber\\
    &= \mathbf{z}^{T}_{\mathbf p} \mathbf{D_p(z_{2p}}+\alpha_{\mathbf p}-\dot{\mathbf p}_{\mathbf d}) + \mathbf{z}^{T}_{\mathbf{2p}}\mathbf{D_{2p}}(\bar{\nu}_{\mathbf p}+\Delta \nu_{\mathbf p}+\eta_{\mathbf p}-\dot{\alpha}_{\mathbf p}) \nonumber\\
    &= -\mathbf{z}^{T}_{\mathbf p} \Lambda_{1p}\mathbf{D_p {z}_p}-\mathbf{z}^{T}_{\mathbf{2p}}\Lambda_{2p}\mathbf{D_{2p}}\mathbf{{z}_{2p}}+\mathbf{z}^{T}_{\mathbf{2p}}\mathbf{D_{2p}}(\Delta\mathbf{{\mathbf{\nu_p}+\eta_p})}, \label{eq:lyap_de_p} \\
    &\dot{V}_q = \sum_{i=1}^{3}\frac{z_{q_i}^{T}\dot{z}_{q_i}}{k_{q_i}^{2}-z_{q_i}^{2}} + \frac{z_{2q_i}^{T}\dot{z}_{2q_i}}{k_{2q_i}^{2}-z_{2q_i}^{2}}= \mathbf{z}^{T}_{\mathbf{q}}\mathbf{D_q}\dot{\mathbf z}_{\mathbf q} + \mathbf{z}^{T}_{\mathbf{2q}} \mathbf{D_{2q}}\dot{\mathbf z}_{\mathbf {2q}} \nonumber\\
    &= \mathbf{z}^{T}_{\mathbf{q}}\mathbf{D_q(z_{2q}+\alpha_q}-\dot{\mathbf{q}}_\mathbf{d}) + \mathbf{z}^{T}_{\mathbf{2q}}\mathbf{D_{2q}}(\bar{\nu}_{\mathbf q}+\Delta\mathbf{\nu_q+\mathbf{\eta_q}}-\dot{\alpha}_{\mathbf q}) \nonumber\\
    &= -\mathbf{z}^{T}_{\mathbf{q}}\Lambda_{1q}\mathbf{D_q}{\mathbf z}_{\mathbf q}-\mathbf{z}^{T}_{\mathbf{2q}}\Lambda_{2q}\mathbf{D_{2q}}{\mathbf z}_{\mathbf{2q}}+\mathbf{z}^{T}_{\mathbf{2q}}\mathbf{D_{2q}}(\Delta\mathbf{{\mathbf{\nu_q}+\eta_q})}. \label{eq:lyap_de_q}
\end{align}
Further simplifications from (\ref{eq:lyap_de_p})-(\ref{eq:lyap_de_q}) yield
\begin{align}
    \dot{V}\leq & -\sum_{i=1}^{3}\gamma_{p_i}\log{\left(\ddfrac{k_{p_i}^{2}}{k_{p_i}^2-z_{p_i}^2}\right)} -\gamma_{q_i}\log{\left(\ddfrac{k_{q_i}^{2}}{k_{q_i}^2-z_{q_i}^2}\right)} \nonumber \\
    &-\gamma_{2q_i}\log{\left(\ddfrac{k_{2q_i}^{2}}{k_{2q_i}^2-z_{2q_i}^2}\right)} -\gamma_{2q_i}\log{\left(\ddfrac{k_{2q_i}^{2}}{k_{2q_i}^2-z_{2q_i}^2}\right)} \\
    & +\mathbf{z}^{T}_{\mathbf{2p}}\mathbf{D_{2p}}(\Delta\mathbf{\nu_p+\eta_p})
    +\mathbf{z}^{T}_{\mathbf{2q}}\mathbf{D_{2q}}(\Delta\mathbf{\nu_q+\eta_q}). \label{eq:lyap_der_inequality}
\end{align}
The first four terms in the aforementioned inequality are derived from the fact that $\log{\left(\ddfrac{k_x^{2}}{k_x^{2}-x^{2}}\right)} <   \ddfrac{x^{2}}{k_x^{2}-x^{2}}$ within any compact set $\Omega: |x(t)|<k_x$ $\forall t \geq 0$ and any $k_x \in \mathbb{R}^{+}$ \cite{Liu2016Robust_BLF}. 

The definition of $V$ yields
\begin{align}
    V &\leq \sum_{i=1}^{3}\log{\left(\ddfrac{k_{p_i}^{2}}{k_{p_i}^{2}-z_{p_i}^{2}}\right)} + \log{\left(\ddfrac{k_{q_i}^{2}}{k_{q_i}^{2}-z_{q_i}^{2}}\right)} + \log{\left(\ddfrac{k_{2p_i}^{2}}{k_{2p_i}^{2}-z_{2p_i}^{2}}\right)} \nonumber\\
    &+ \log{\left(\ddfrac{k_{2q_i}^{2}}{k_{2q_i}^{2}-z_{2q_i}^{2}}\right)}. \label{eq:lyap_ineq}
\end{align}
Using \eqref{eq:lyap_ineq}, from \eqref{eq:lyap_der_inequality} we have
\begin{align}
    \dot{V} \leq - \varsigma V + \mathbf{z}^{T}_{\mathbf{2p}}\mathbf{D_{2p}} (\Delta \mathbf{\nu_p+\eta_p})  + \mathbf{z}^{T}_{\mathbf{2q}}\mathbf{D_{2q}}(\Delta \mathbf{\nu_q+\eta_q}) \label{eq:lyap_der_varsigma}
\end{align}
where $\varsigma \triangleq \min \lbrace \min\lbrace \gamma_{p_i} \rbrace, \min\lbrace \gamma_{q_i} \rbrace,\min \lbrace \gamma_{2p_i} \rbrace, \min\{\gamma_{2q_i}\}\rbrace$.

Based on the structures of $\Delta \mathbf{\nu_{p}}$ and $\Delta \mathbf{\nu_{q}}$ as in (\ref{eq:delta_nu_p}) and (\ref{eq:delta_nu_q}), the following four possible cases can be identified.\\
\noindent \textbf{Case (i)} $\|\mathbf{z_{2p}}\|\geq \epsilon_p$ and $\|\mathbf{z_{2q}}\|\geq \epsilon_q$:
\\
Since $\rho_p\geq\|\mathbf{\eta_p}\|$ and $\rho_q\geq\|\mathbf{\eta_q}\|$ by design, \eqref{eq:lyap_der_varsigma} yields
\begin{align}
    \dot{V} \leq &- \varsigma V - \rho_p\|\mathbf{D_{2p}}\|\|\mathbf{z_{2p}}\|  + \rho_p\|\mathbf{D_{2p}}\|\|\mathbf{z_{2p}}\| \nonumber\\
    &-\rho_q\|\mathbf{D_{2q}}\|\|\mathbf{z_{2q}}\| +\rho_q\|\mathbf{D_{2q}}\|\|\mathbf{z_{2q}}\| \nonumber \\
    \implies \dot{V} \leq & -\varsigma V \label{eq:case1}
\end{align}
\noindent \textbf{Case (ii)} $\|\mathbf{z_{2p}}\|\geq \epsilon_p$, $\|\mathbf{z_{2q}}\| < \epsilon_q$:
\begin{align}
    \dot{V} &\leq - \varsigma V + \|\mathbf{\eta_q}\|\|\mathbf{D_{2q}}\|\|\mathbf{z_{2q}}\|
    \leq - \varsigma V + \|\mathbf{\eta_q}\|\|\mathbf{D_{2q}}\|\epsilon_q
\end{align}
\noindent \textbf{Case (iii)} $\|\mathbf{z_{2p}}\| < \epsilon_p$, $\|\mathbf{z_{2q}}\| \geq \epsilon_q$:
\begin{align}
    \dot{V} &\leq - \varsigma V + \|\mathbf{\eta_p}\|\|\mathbf{D_{2p}}\|\|\mathbf{z_{2p}}\| 
    \leq - \varsigma V + \|\mathbf{\eta_p}\|\|\mathbf{D_{2p}}\|\epsilon_p
\end{align}
\noindent \textbf{Case (iv)} $\|\mathbf{z_{2p}}\| < \epsilon_p$, $\|\mathbf{z_{2q}}\| < \epsilon_q$:
\begin{align}
    \dot{V} & \leq - \varsigma V + \|\mathbf{\eta_m}\|\|\mathbf{D_{2m}}\|\epsilon_m
\end{align}
where $|| \mathbf{\eta_m}||\triangleq \max \lbrace ||\mathbf{\eta_p}||, ||\mathbf{\eta_q}|| \rbrace$,\\ $|| \mathbf{D_{2m}}||\triangleq \max \lbrace ||\mathbf{D_{2p}}||, || \mathbf{D_{2q}} || \rbrace$, and $\epsilon_m \triangleq \max \lbrace \epsilon_p, \epsilon_q \rbrace$. Replacing (\ref{alpha_p}) into (\ref{eq:z_2p}) and (\ref{alpha_q}) into (\ref{eq:z_2q}) one can verify that $\mathbf{z_{2p}}, \mathbf{z_{2q}} \in \mathcal{L}_{\infty} \Rightarrow \mathbf{p}, \dot{\mathbf{p}}, \mathbf{q}, \dot{\mathbf{q}} \in \mathcal{L}_{\infty}$ as desired trajectories are bounded via Assumption 3. Now, boundedness of $\mathbf{p}, \dot{\mathbf{p}}, \mathbf{q}, \dot{\mathbf{q}}$ imply $\alpha_{\mathbf{p}}, \alpha_{\mathbf{q}} \in \mathcal{L}_{\infty} \Rightarrow\mathbf{\nu_{p}}, \mathbf{\nu_{q}}\in \mathcal{L}_{\infty}$ from (\ref{ctr_p}), (\ref{ctr_q}) and $\mathbf{C(q,\dot{q})} \in \mathcal{L}_{\infty}$ by property of EL mechanics \cite{spong2008robot}; these cumulative boundedness conditions imply $\mathbf{\eta_{p}},\mathbf{\eta_{q}}, \mathbf{D_{2p}}, \mathbf{D_{2q}}\in \mathcal{L}_{\infty}$ from (\ref{un_p}) and (\ref{uncer_att}). Therefore, there exists a finite constant $c$ such that $\|\mathbf{\eta_m}\|\|\mathbf{D_{2m}}\|\epsilon_m \leq c$. Then, observing the four stability cases it can be inferred that
\begin{align}
    \dot{V} \leq - \varsigma V + c
\end{align}
implying that the closed-loop system remains bounded via the barrier Lyapunov function $V$, implying $z_{pi}, z_{qi}, \dot{z}_{pi}, \dot{z}_{qi}$ never violates the constraints and  $|z_{pi}| < k_{pi}, |z_{qi}| < k_{qi},|\dot{z}_{pi}| < \dot{k}_{pi}, |\dot{z}_{qi}| < \dot{k}_{qi} $ $i=1,2,3$ $\forall t>0$. 
\end{proof}
\section{Experimental Verification}

\begin{figure*}[ht!]
   \includegraphics[width=6in, height=3.7in]{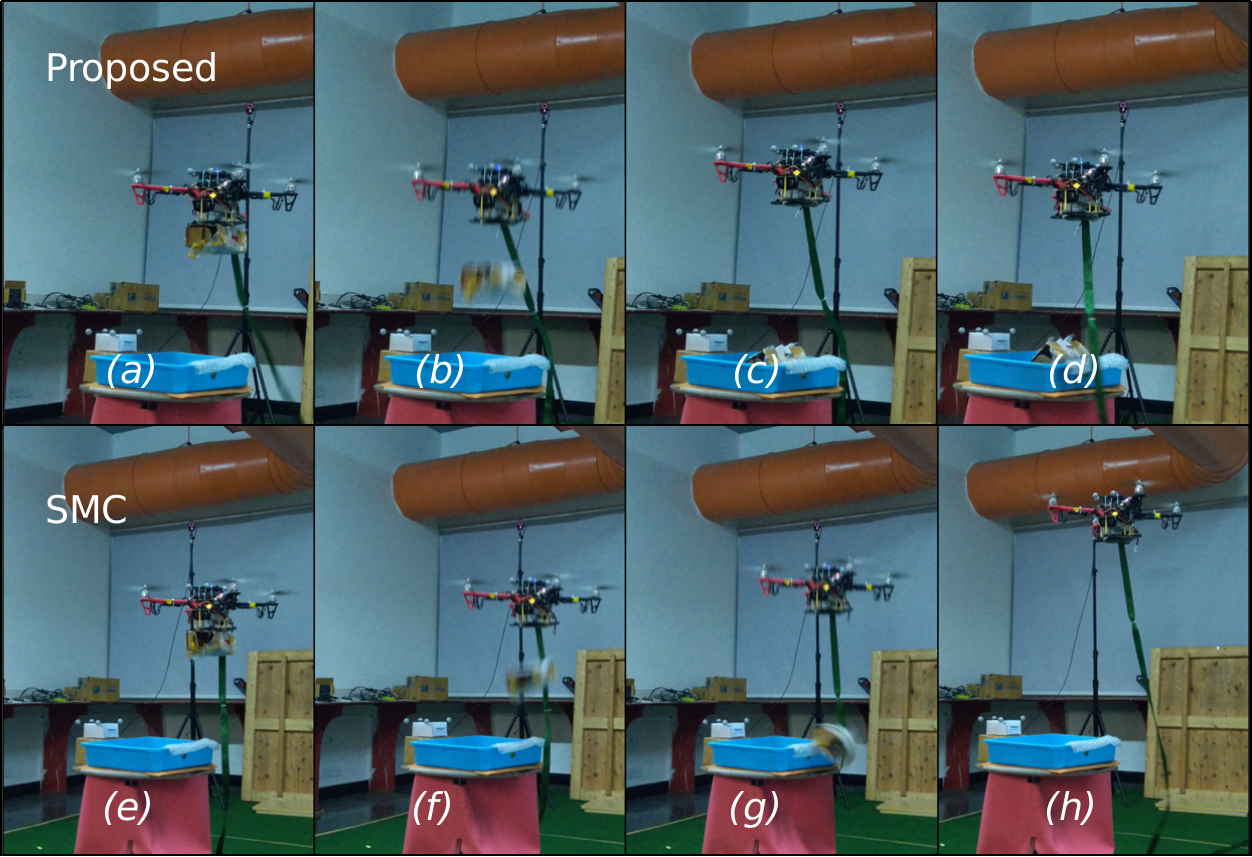}
    \centering
   \caption{Snapshots from the experiment: (a), (e) Before dropping payload on drop-point (b), (f) After dropping payload (t=37s) (c), (g) Payload lands on the tray for proposed controller, while it falls on the edge for SMC (d), (h) Proposed controller makes quadrotor hold its position, while SMC shoots the quadrotor up}
     \label{fig:exp_snaps}
\end{figure*}

This section experimentally verifies the effectiveness of the proposed controller for precision maneuvering in comparison with SMC (\cite{xu2008sliding})\footnote{Note that the standard BLF-based controller for quadrotors \cite{dasgupta2019singularity, kumar2020barrier} are not robust to uncertainties; therefore, in the face of parametric uncertainty (attachment and release of payload) and external disturbances, tracking error might approach the constraints leading to high control input and potential operational hazard by crashing the quadrotor. Therefore, real-time experiments with standard non-robust BLF methods were avoided.}.   
\newpage
\subsection{Experimental setup and scenario}
A quadrotor (1.4 kg) is tasked follow a ground robot (Pioneer 3-DX) and drop a payload on a tray placed on top of it (cf. Fig. \ref{fig:exp_snaps}). 
The quadrotor setup includes one raspberry pi-4 processing unit and one electromagnetic gripper ($0.03$ kg). Optitrack motion-capture system (at 60 fps) is used to obtain quadrotor's pose. Note that the gripper operations are carried out via a remote signal, which is separate from the control design.
\begin{figure}[hbt]
    \includegraphics[width=5.5in,height=4in]{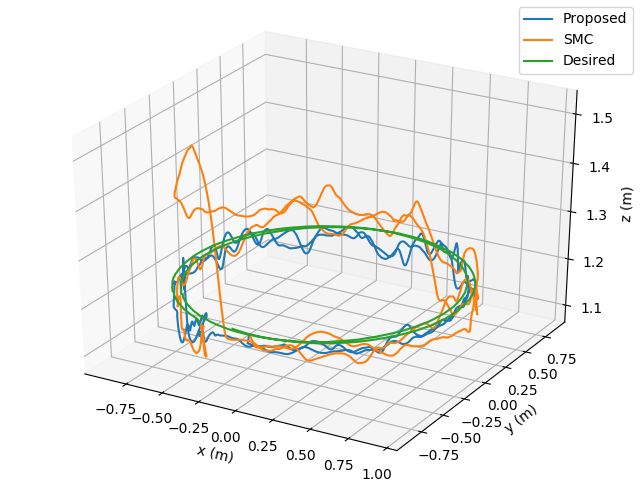}
   \centering
   \caption{Circular path tracking performance comparison.}
     \label{fig:3d_plot}   
 \end{figure}
\begin{figure}[ht!]
    \includegraphics[width=5in,height=3.5in]{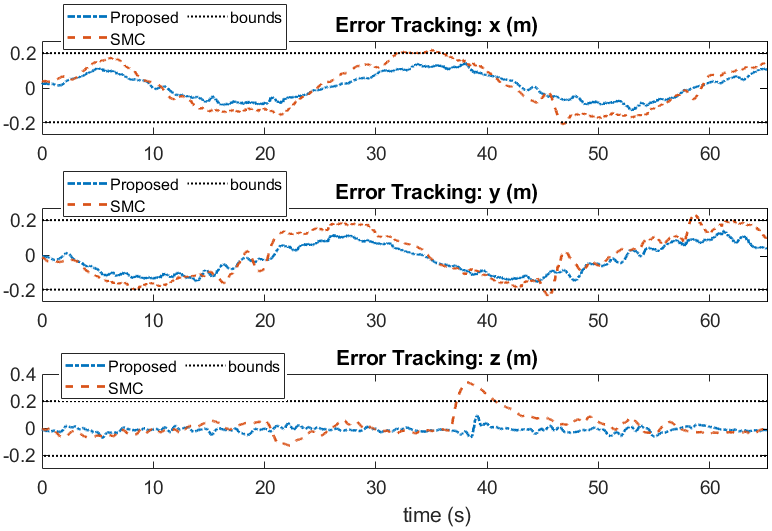}
    \centering
    \caption{Position tracking error comparison.}
    \label{fig:exp_pos_error}   
\end{figure}
\begin{figure}[ht!]
    \includegraphics[width=5in,height=3.5in]{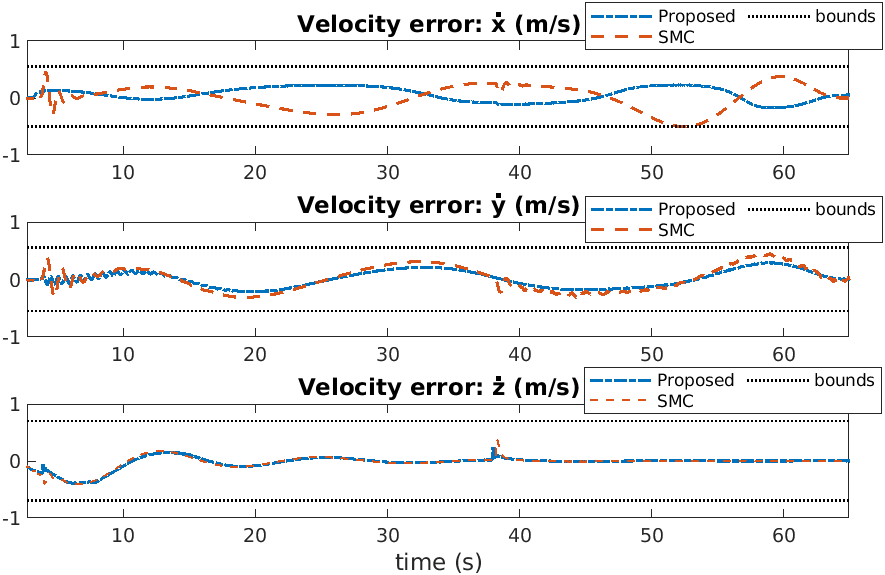}
    \centering
    \caption{Velocity tracking error comparison.}
    \label{fig:exp_vel_error}   
\end{figure}
\begin{figure}[ht!]
    \includegraphics[width=5in,height=3.5in]{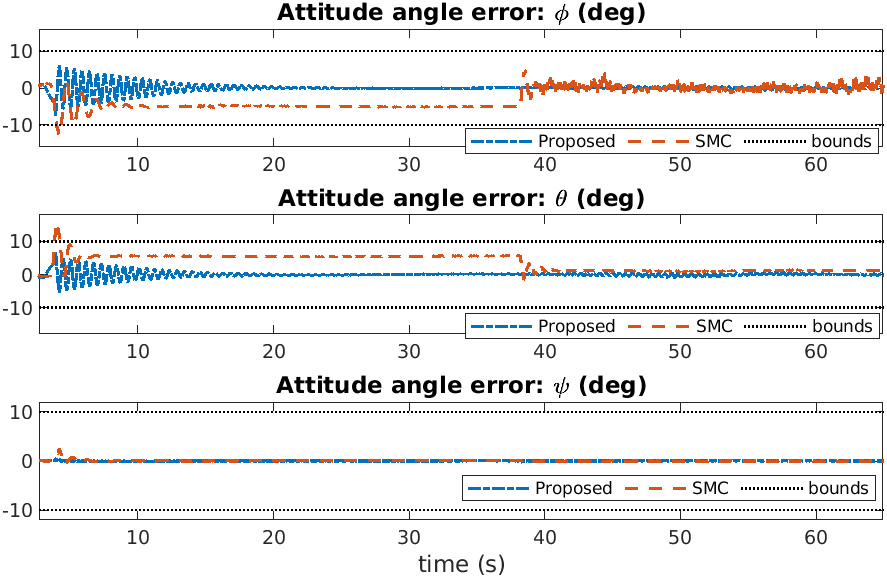}
    \centering
    \caption{Attitude tracking error comparison.}
    \label{fig:exp_att_error}
\end{figure}
\begin{figure}[ht!]
    \includegraphics[width=5in,height=3.5in]{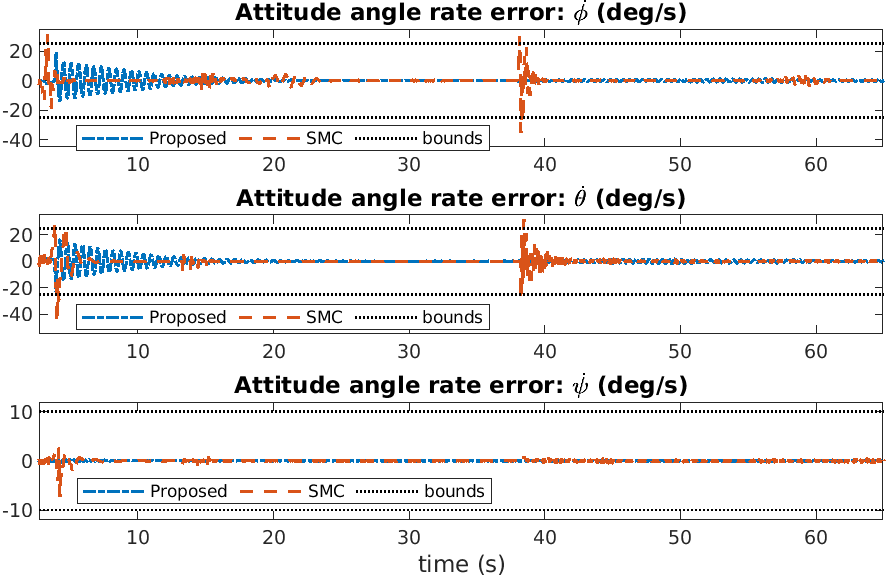}
    \centering
    \caption{Attitude-rate tracking error comparison.}
    \label{fig:exp_attrate_error}
\end{figure}
\FloatBarrier
\begin{figure}[ht!]
    \includegraphics[width=5in,height=3.5in]{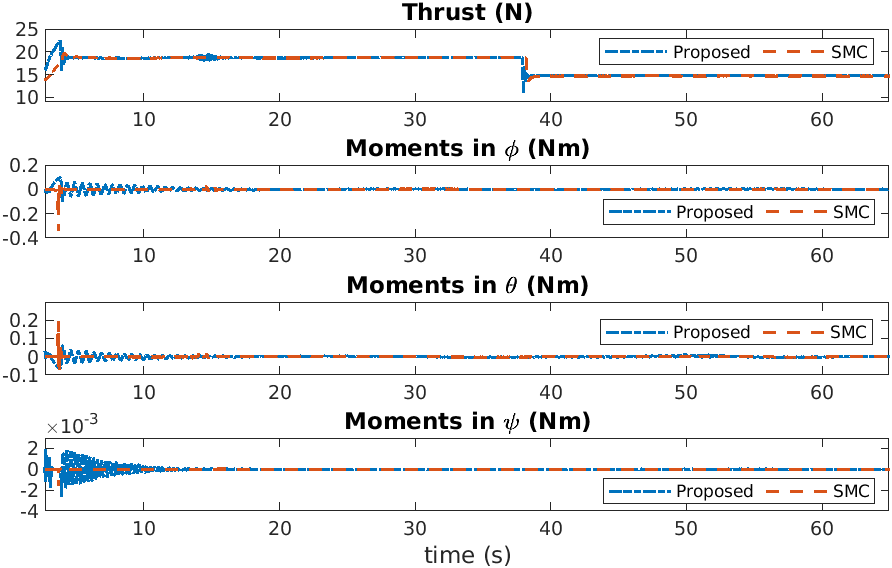}
    \centering
    \caption{Control signals comparison.}
    \label{fig:exp_control}   
\end{figure}
The experimental scenario is described as follows: (i) Pioneer 3-DX  is commanded to make two rounds of a circular path of radius $1$m; (ii) the $x-y$ coordinates of Pioneer 3-DX from Optitrack system are used as the desired setpoints for the quadrotor with a desired constant altitude of $z_d(t)=1.2$m, so that it flies over the ground robot all time ($\psi_d=0$, and $\phi_d, \theta_d$ follow Remark 3 of Chapter 3); (iii) the quadrotor initially carries a payload of $0.3$ kg for one complete round and drops it at $t=37$s; (iv) then, it tracks the ground robot for another turn without the payload. 

Control parameters are selected as: constraint parameters $[k_{p_1}, k_{p_2}, k_{p_3}]= [0.2, 0.2, 0.2]$~m (according to the size of the tray), $[\dot{k}_{p_1}, \dot{k}_{p_2}, \dot{k}_{p_3}]= [0.55, 0.55, 0.7]$~m/s, $ [k_{q_1}, k_{q_2}, k_{q_3}]=[ 0.174, 0.174, 0.174]$~rad, and $[\dot{k}_{q_1}, \dot{k}_{q_2}, \dot{k}_{q_3}]=[ 0.43, 0.43, 0.174]$ rad/s; $\Lambda_{1p}=\diag \lbrace 0.5, 0.5, 0.5 \rbrace $, $\Lambda_{2p} = \diag{\{10, 10, 10\}}$, $\Lambda_{1q}=\diag \lbrace 20, 20, 12 \rbrace$, $\Lambda_{2q} = \diag{\{15, 15, 15\}}$, $\bar{m}=1.5$, \\
$\bar{\mathbf{J}}=\diag{\{0.02, 0.02, 0.04\}}$, $E_p=E_q = 0.3$, $\epsilon_p = 0.1$, $\epsilon_q = 1$, and the upper-bounds of $\|\mathbf{d_p}\|$ and $\|\mathbf{d_q}\|$ as 1.73 and 0.173, respectively. For SMC, sliding surfaces are selected as:  $\mathbf{s_p}=\dot{\mathbf{z}}_{\mathbf{p}} + \Lambda_{1p} \mathbf{z_p}$, $\mathbf{s_q}=\dot{\mathbf{z}}_{\mathbf{q}} + \Lambda_{1q} \mathbf{z_q}$.
\newpage
\subsection{Experimental results and analysis}
\begin{table}[htb!]
\begin{center}
\renewcommand{\arraystretch}{1.6}
		\centering
{
{	\begin{tabular}{|c| c c c c c c|}
		
	\hline
	Controller	& \multicolumn{3}{c}{RMS error (m)} & \multicolumn{3}{c|}{RMS error (degree)}  \\ \cline{1-7}
		  & $x$ & $y$  & $z$  & $\phi$ & $\theta$  & $\psi$  \\
		\hline
		SMC & 0.12 & 0.12  & 0.09 & 3.84 & 4.22 & 0.17 \\
		\hline
		Proposed & {0.08} & {0.08}  & {0.02} & 1.08 & 0.90 & 0.03 \\
		\hline
& \multicolumn{3}{c}{Peak absolute error (m)} & \multicolumn{3}{c|}{Peak absolute error (degree)}  \\ \cline{1-7}
		 & $x$ & $y$  & $z$  & $\phi$ & $\theta$  & $\psi$  \\
		\hline
		SMC & 0.22 & 0.24  & 0.34 & 12.75 & 15.72 & 2.53 \\
		\hline
		Proposed & {0.15} & {0.16}  & {0.09} & 7.59 & 6.66  & 0.46 \\
		\hline
\end{tabular}}}
\caption{{Position tracking performance comparison}}
\label{table 3}
\end{center}
\end{table}
\begin{table}[htb!]
\renewcommand{\arraystretch}{1.6}
		\centering
{
{	\begin{tabular}{|c| c c c c c c|}
    \hline
	Controller	& \multicolumn{3}{c}{RMS error (m/s)} & \multicolumn{3}{c|}{RMS error (degree/s)}  \\ \cline{1-7}
		  & $\dot{x}$ & $\dot{y}$  & $\dot{z}$  & $\dot{\phi}$ & $\dot{\theta}$  & $\dot{\psi}$  \\
		\hline
		SMC & 0.21 & 0.20  & 0.11 & 2.61 & 3.57  & 0.40 \\
		\hline
		Proposed & {0.19} & {0.13}  & {0.08} & {2.48} & 2.69 & 0.05 \\
		\hline
& \multicolumn{3}{c}{Peak absolute error (m/s)} & \multicolumn{3}{c|}{Peak absolute error (degree/s)}  \\ \cline{1-7}
		 & $\dot{x}$ & $\dot{y}$  & $\dot{z}$  & $\dot{\phi}$ & $\dot{\theta}$  & $\dot{\psi}$  \\
		\hline
		SMC & 0.50 & 0.46  & 0.37 & 35.51 & 46.61  & 7.31 \\
		\hline
		Proposed & {0.26} & {0.29}  & {0.36} & 19.04 & 21.12  & 0.54 \\
		\hline
\end{tabular}}}
\caption{{Velocity tracking performance comparison}}
\label{table 2_chap_4}
\end{table}
\FloatBarrier
The performance comparison of the proposed controller and SMC are highlighted via Figs. \ref{fig:3d_plot} - \ref{fig:exp_control} and via Tables \ref{table 3}-\ref{table 2_chap_4} (in terms of root-mean-squared (RMS) and absolute peak error). Figures \ref{fig:exp_pos_error}-\ref{fig:exp_attrate_error} reveal that SMC breached the bounds in $(x, y, z, \phi, \theta)$ and in ($\dot{\phi}$, $\dot{\theta}$), while it was too close to the bounds in $\dot{x}, \dot{y}$: these resulted in the quadrotor dropping the payload on the edge of the tray and eventually falling outside the tray (cf. snapshots (g) and (h) in Fig. \ref{fig:exp_snaps}) in the case of SMC;
whereas, the proposed controller could maintain the errors of all controlled states (position, orientation, linear and angular velocity) well within the bounds and the payload was dropped within the tray (cf. snapshot (c) and (d) in Fig. \ref{fig:exp_snaps}). A sudden deviation in the altitude can be observed for SMC at $t=37$s in snapshot (h) of Fig. \ref{fig:exp_snaps} and in Figs. \ref{fig:3d_plot} and \ref{fig:exp_pos_error} when the payload is dropped. These transients happen because higher thrust was required while carrying the payload to maintain the altitude compared to the no payload condition (cf. Fig. \ref{fig:exp_control}); just after releasing the payload, this additional thrust tends to push the quadrotor upwards. Left unattended, this would be hazardous in confined spaces. Whereas, the proposed controller could maintain the altitude within the bound (cf. snapshot (d) of Fig. \ref{fig:exp_snaps} and Figs. \ref{fig:3d_plot} and \ref{fig:exp_pos_error}).

\section{Conclusion}
A robust controller with Barrier Lyapunov approach for quadrotors was formulated to ensure tracking performance with full state-constraints under parametric uncertainties and external disturbances. Closed-loop system stability was established analytically and the performance of the proposed controller was experimentally verified. 


\chapter{Conclusions}
\label{ch:conc}
In this thesis, comprehensive robust control solutions for a quadrotor manoeuvring under multiple state-constraints was provided. For the constraint handling, the Barrier Lyapunov Function (BLF) method was followed. Chapter 1 detailed the challenges faced by quadrotor tasked with manoeuvring operations in tight spaces which acted as strong motivation behind the thesis construction, and introduced fundamental principles that were utilized in the chapters that followed. Chapter 2 successfully demonstrated the effectiveness of the BLF based control method over unconstrained methods such as the PID controller, by designing a BLF controller on the generalized n-DoF Euler-Lagrange dynamics. Chapter 3 utilized the potential of the BLF control method and successfully designed a robust controller with user-specified constraints on the position and orientation of a quadrotor under parametric uncertainties and external disturbances; the efficacy of the proposed control method was validated via extensive realistic simulation scenarios on the Gazebo platform where the results obtained were compared with Sliding Mode Control (SMC). Lastly, Chapter 4 designed a robust controller with full state-constraints (i.e., constraints on the position, attitude, linear velocity and angular velocity) on a quadrotor under uncertainties and disturbances. The performance of the proposed controller was validated via experiments on a real quadrotor carrying payload, where the proposed controller outperformed SMC. 

Throughout the thesis, the initial trajectory errors must lie within the imposed bounds for the feasibility of the control method. This opens up opportunities to explore the implementation of variable constraints on a quadrotor, where the initial values of the bounds could be made greater than their steady state values, increasing the window of feasibility. Additionally, the possibility of implementing constraints on the actuators in BLF based robust control could also be explored. Finally, a challenging future work could be to implement state-constraints on the quadrotor in an adaptive setting with unknown uncertainty bounds.


\chapter*{Related Publications}
\label{ch:relatedPubs}
\section{Main Publications}

\begin{itemize}
    \item S. Ganguly, V. N. Sankaranarayanan, B. V. S. G. Suraj, R. D. Yadav and S. Roy, “Efficient Manoeuvring of Quadrotor under Constrained Space and Predefined Accuracy”, \textit{IEEE/RSJ International Conference on Intelligent Robots and Systems, 2021}.
    
    \item S. Ganguly, V. N. Sankaranarayanan, B. V. S. G. Suraj, R. D. Yadav and S. Roy, “Robust Manoeuvring of Quadrotor under Full State Constraints”, \textit{Advances in Control and Optimization of Dynamical Systems, 2022}. 
\end{itemize}

\section{Other Publication(s)}

\begin{itemize}
    \item V. N. Sankaranarayanan, R. D. Yadav, R. K. Swayampakula, S. Ganguly and S. Roy, “Robustifying Payload Carrying Operations for Quadrotors under Time-varying State Constraints and Uncertainty”, \textit{IEEE Robotics and Automation Letters (RA-L)}.  (Submitted)

\end{itemize}


\bibliographystyle{IEEEtran}
\bibliography{sampleBib} 

\end{document}